\documentclass{sig-alternate-05-2015}

\usepackage{graphicx}
\usepackage{cite}
\usepackage{paralist}
\usepackage[tight,footnotesize]{subfigure}
\usepackage{mathrsfs}
\usepackage{amsmath}
\usepackage{amssymb}
\usepackage[linesnumbered,ruled,vlined]{algorithm2e}

\usepackage{slashbox}

\usepackage{times}

\newcommand{\rom}[1]{\romannumeral #1}
\newcommand{\nop}[1]{}
\makeatletter % `@' now normal "letter"
\@addtoreset{equation}{section}
\makeatother  % `@' is restored as "non-letter"

\begin{document}

% Copyright
%\setcopyright{acmcopyright}
%\setcopyright{acmlicensed}
%\setcopyright{rightsretained}
%\setcopyright{usgov}
%\setcopyright{usgovmixed}
%\setcopyright{cagov}
%\setcopyright{cagovmixed}

% DOI
%\doi{10.475/123_4}

% ISBN
%\isbn{123-4567-24-567/08/06}

%Conference
%\conferenceinfo{PLDI '13}{June 16--19, 2013, Seattle, WA, USA}

%\acmPrice{\$15.00}

%
% --- Author Metadata here ---
%\conferenceinfo{WOODSTOCK}{'97 El Paso, Texas USA}
%\CopyrightYear{2007} % Allows default copyright year (20XX) to be over-ridden - IF NEED BE.
%\crdata{0-12345-67-8/90/01}  % Allows default copyright data (0-89791-88-6/97/05) to be over-ridden - IF NEED BE.
% --- End of Author Metadata ---

\title{Tuning Crowdsourced Human Computation}

\author{\hspace*{3pt}  Chen Cao$^{\S}$ \hspace*{1pt}Zheng Liu $^{\S}$ \hspace*{1pt}
Lei Chen$^{\S}$ \hspace*{1pt} H. V. Jagadish   $^{\dag}$\hspace*{1pt}\\
{$^\S$Hong Kong University of Science and Technology, Hong Kong, China } \\
{$^\dag$University of Michigan, Ann Arbor, MI, USA} \\
{{caochen@cse.ust.hk, zliu@cse.ust.hk, leichen@cse.ust.hk,
jag@umich.edud}}
\\}

\maketitle
\newtheorem{definition}{Definition}
\newtheorem{lemma}{Lemma}
\newtheorem{theorem}{Theorem}
\newtheorem{corollary}{Corollary}
\newtheorem{Example}{Example}
\newtheorem{motivation example}{Motivation Example}
\newtheorem{conjecture}{Conjecture}
\newtheorem{hypothesis}{Hypothesis}
\newcommand{\ud}{\mathrm{d}}
\newcommand{\ue}{\mathrm{e}}
\newcommand*\diff{\mathop{}\!\mathrm{d}}

\begin{abstract}
As the use of crowdsourcing increases, it is important to think about performance optimization. For this purpose, it is possible to think about each worker as a  \textbf{HPU}(\textit{Human Processing Unit}\cite{davis:hpu10}), and to draw inspiration from performance optimization on traditional computers or cloud nodes with CPUs. However, as we characterize HPUs in detail for this purpose, we find that there are important differences between CPUs and HPUs, leading to the need for completely new optimization algorithms.

In this paper, we study the specific optimization problem of obtaining results fastest for a crowd sourced job with a fixed total budget. In crowdsourcing, jobs are usually broken down into sets of small tasks, which are assigned to workers one at a time.  We consider three scenarios of increasing complexity: \emph{Identical Round Homogeneous} tasks, \emph{Multiplex Round Homogeneous} tasks, and \emph{Multiple Round Heterogeneous} tasks. For each scenario, we analyze the stochastic behavior of the HPU clock-rate as a function of the remuneration offered. After that, we develop an optimum Budget Allocation strategy to minimize the latency for job completion. We validate our results through extensive simulations and experiments on Amazon Mechanical Turk.
\end{abstract}


\section{Introduction}\label{section:intro}

Human Computation \cite{law:human11} has emerged in recent years as
a new and exciting compute paradigm. As a powerful compement of
traditional computer systems, human computation naturally allows
tasks with human-intrinsic values or features, like comparing
emotions of speeches, identifying objects in images and so on.  The
emergence of public crowdsourcing platforms, which provide a
scalable manageable workforce resource, has boosted the utilization
of this long-discovered\cite{grier:ieeeann98} human cognitive
ability.  A wide range of data-driven applications now benefit from
human computation by considering it as a new computing component.
Examples include
\begin{inparaenum}[\itshape a\upshape)]
  \item crowd-powered databases\cite{aditya:scoop11,michael:sigmod11:crowddb,marcus:vldb11} and fundamental operators like filtering\cite{aditya:sigmod12} and Max\cite{venetis:www12,guo:sigmod12}, group-by\cite{milo:icdt13},
  \item advanced data processing technologies like image tagging\cite{reynold:icde13:tagging}, schema matching\cite{jason:vldb13} and entity resolution\cite{jiannan:sigmod13}, and
  \item combinatorial problems like planning\cite{milo:icde13} and mining\cite{milo:sigmod13}.
\end{inparaenum}

\begin{figure}[htbp] \centering
\subfigure[Example 1: Repetition] { \label{fig:eg1}
\includegraphics[height = 1.05in,width=0.45\columnwidth]{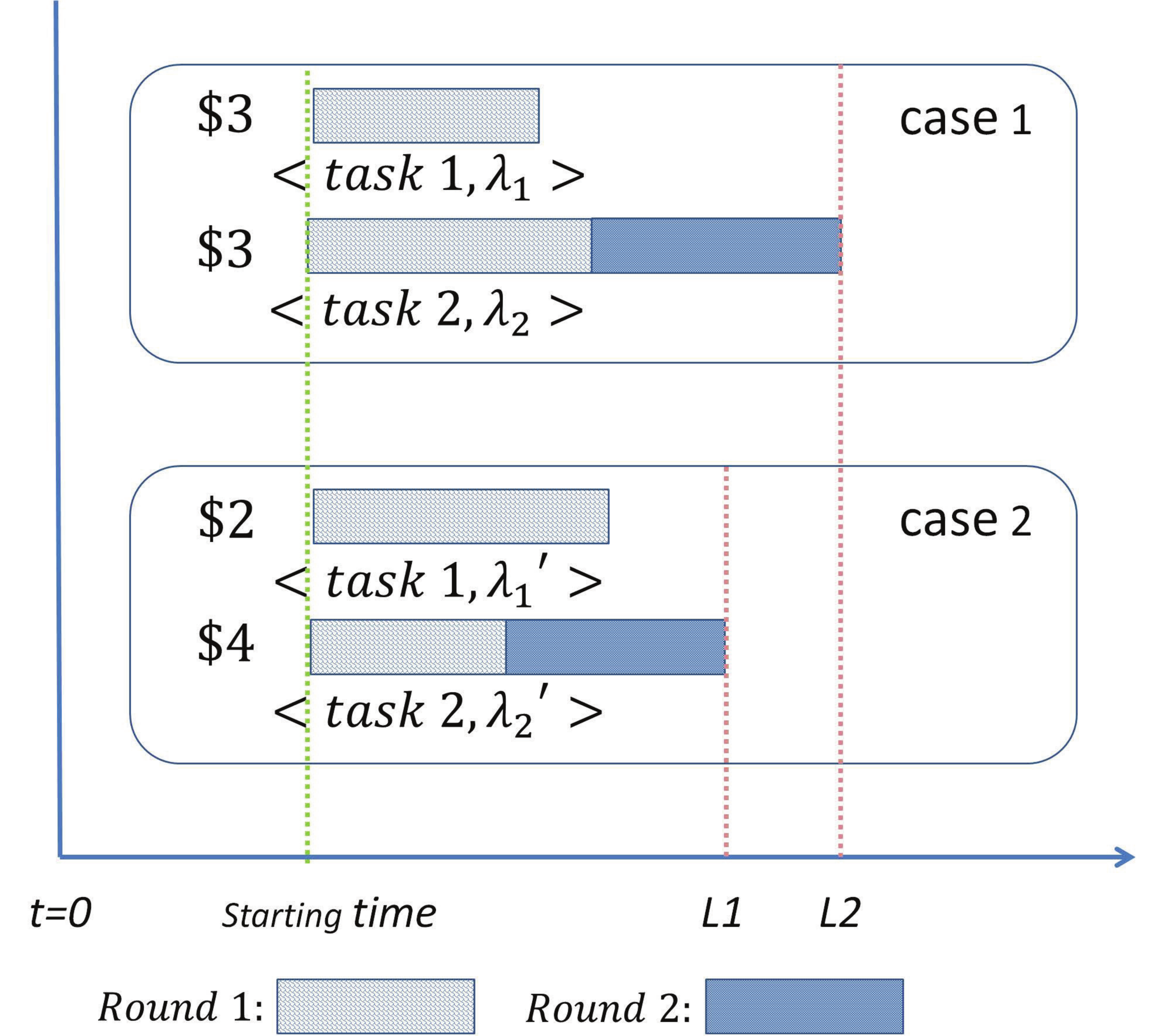}
}
\subfigure[Example 2: Heterogeneous] { \label{fig:eg2}
\includegraphics[height = 1.05in,width=0.45\columnwidth]{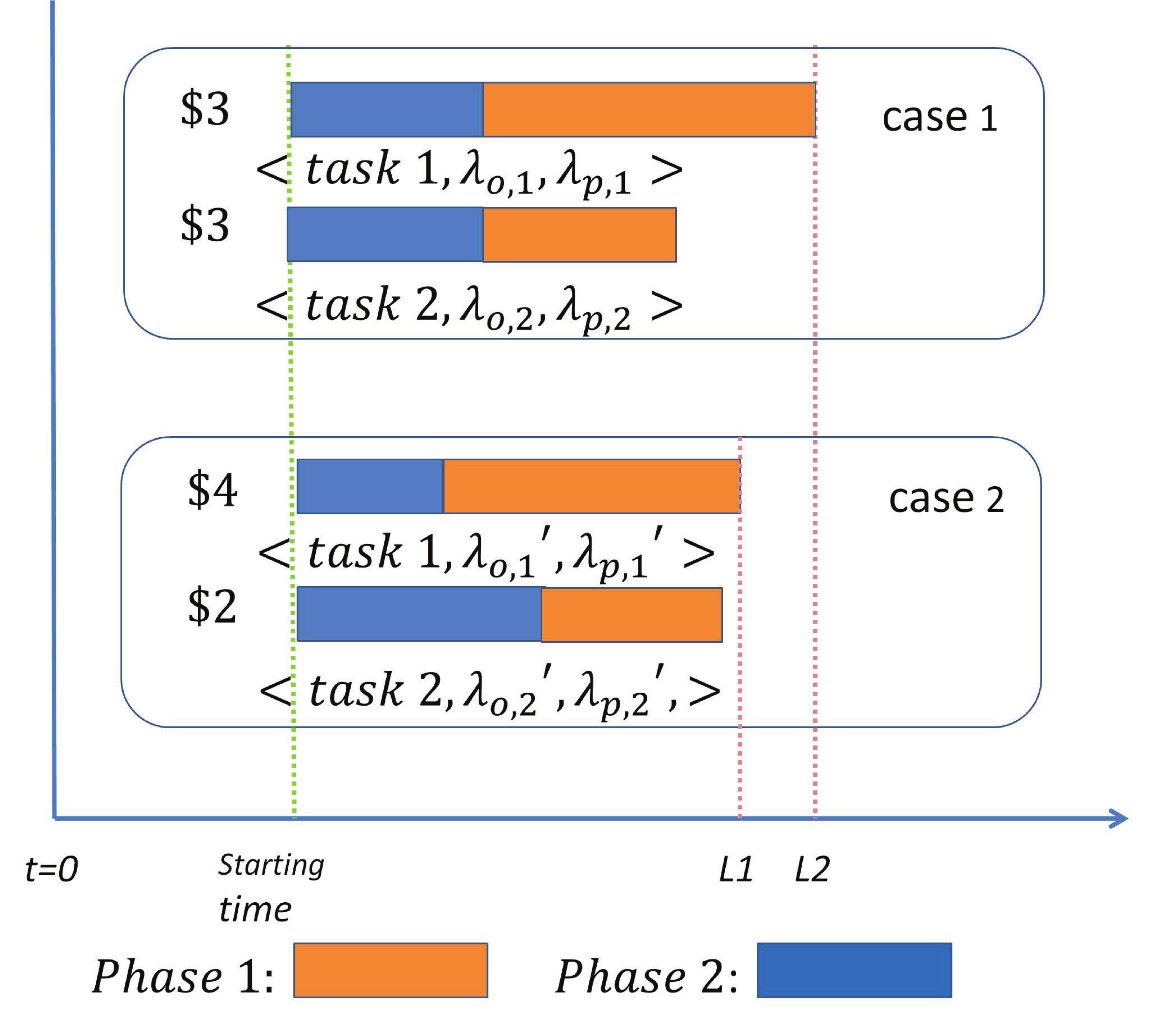}
}
\caption{Demonstration of Motivation
Example}\label{fig:example}
\end{figure}

%Furthermore, resembling to CPU-based system, complex operations can
%be decomposed into basic human operations which are termed as atomic
%tasks. For example,  we can decompose  filtering and schema matching
%into \textit{yes-no} voting and ranking into Borda
%Voting\cite{voting:jstor73}.

As crowdsourcing becomes more prevalent, there is an effort to
understand and characterize it better.  In this regard, it has been
suggested that the system can be viewed as comprising Human
Processing Units (HPUs) that are analogous to CPUs of traditional
computers.  As atomic task performed by a worker is then one
``instruction'' of the HPU, and the time to respond is the ``clock
cycle''. However, the HPU has many characteristics that differ from
those of a CPU:

\begin{inparaenum}[\itshape i\upshape)]
  \item the clock time is \textbf{stochastic};
  \item the results are \textbf{error-prone} according to a probability;
 % \item the basic operation is \textbf{voting}(making choice from options) with varying cardinality;
  \item the cost includes \textbf{monetary} expense.
\end{inparaenum}
%(Please refer to Table~\ref{table:cpu} as a succinct summary).

Given the HPU abstractions, one can consider optimizing many aspects
of HPU processing. Our focus, in this paper, is the HPU clock rate.
This is because we want to minimize the total latency of a
computational task by optimizing HPU clock rate.

%\begin{table} \label{table:cpu}
%\caption{Comparison between CPU and HPU}
%\centering
%\begin{tabular}{|c|c|c|c|c|c|}
%  \hline
%  Type & Basis & Latency & Quality & Operation & Cost\\
%  \hline
%  \textbf{CPU} & p-n junction & fixed & certain & bit ops & none\\
%  \hline
%  \textbf{HPU} & market & stochastic & error-prone & voting & monetary\\
%  \hline
%\end{tabular}
%\end{table}

The clock rate for the HPU is variable, and is different for each instruction and each instance.
In a crowdsourcing workforce market, a
task is exposed on the market with a promised reward, and then
``workers'' select the task to work on according to their interests.
Recent studies on Amazon Mechanical Turk (AMT) report that the
task acceptance duration follows an exponential
distribution\cite{jwang:csdm11,faridani:hcomp11}, and the rate
$\lambda_{i}$ is mostly determined by the promised
reward\cite{mason:hcomp09} and the type(difficulty) of the
task\cite{faridani:hcomp11}. Then, the processing time of a task
follows an exponential distribution with another
rate\cite{yan:mobisys10}, which is independent of the promised
payment\cite{mason:hcomp09}.

Once the task has been specified, the  only task-owner input that
can control the completion time is the \emph{payment}. If we only
have one task to be performed by HPUs, the solution is very simple
-- the more we can afford to pay, the faster the task will be
completed.

Of course, the computational job at hand is typically performed with
the aid of many HPU tasks
\cite{aditya:scoop11,michael:sigmod11:crowddb,marcus:vldb11}.  In
fact, a typical algorithm architecture repeats each task in multiple
times.  Thus, the requester issues a large number of HPU tasks in
parallel, each possibly to be repeated, and then waits for all HPUs
to return result.  As such, there is limited value to optimize the
clock rate of a single HPU in isolation: what really matters is the
latency of the entire computation, which is determined by the
longest duration among the set of parallel repeated tasks. Thus, our
optimizing HPU clock rate problem is in fact on studying how to
allocate a given fixed budget $B$ in a manner that minimizes the
total latency of a computational job involving HPUs.

%Thus, our problem becomes to allocate a given fixed budget $B$ in a manner that minimizes the total latency of a computational job involving HPU tasks. Such a problem is most relevant to the issues discussed in \cite{aditya:vldb15}. However, what distinguishes our work is that the latency is a crowdsourcing task is model with two phase: the on-hold and processing phase, which is more close to the real world semantic.

To demonstrate the challenging issues involved in optimizing HPU
clock rates, let us consider the following two motivating examples,
both based on a crowd-powered data-base, as proposed in
\cite{aditya:scoop11,michael:sigmod11:crowddb,marcus:vldb11,venetis:www12,guo:sigmod12}.

\begin{motivation example}(Figure~\ref{fig:eg1})
Consider a \textit{sorting} task on 4 given items $O=\{o_{1}, o_{2},
o_{3}, o_{4}\}$. According to the user's requirements, the query
planner, for example the \textit{``next votes''} proposed in
\cite{guo:sigmod12}, decomposes the sorting task into atomic
pairwise voting tasks $T=\{\{o_{1},o_{2}\}\times 1,
\{o_{3},o_{4}\}\times 2\}$, which means the HPU is expected to run
the task of comparison on such pairs for 1 and 2 repetitions (times)
respectively. As illustrated in Figure~\ref{fig:eg1}, the two tasks
commence at the same time, but in order to finalize the entire
query, the database has to wait until the end of the longest atomic
task. There are many choices for budget allocation.  Two obvious
ones are: one is evenly divided to two tasks, 3 for task 1 and 3 for
task 2(case 1); another one is more load-sensitive, 2 for task 1 and
4 for task 2(case 2). The results for the two cases are shown in the
figure, suggesting that the second option is better. But how could
we predict this? Moreover, even if this is the better of these two
choices, is it the best?  What is the best allocation across the two
tasks?
\end{motivation example}

\begin{motivation example}(Figure~\ref{fig:eg2}) Consider now a more complex scenario in which
the database is required to process two types of queries
simultaneously, sorting and filtering\cite{aditya:sigmod12}, where
the latter can also be decomposed into pairwise voting
tasks(\textit{yes} or \textit{no} voting). Suppose two tasks are
given $T=\{\{o_{1},o_{2}\}\times 1,
\{o_{3},\textrm{yes}?\textrm{no}\}\times 1\}$. However, unlike the
previous case, different types of task present different difficulty
levels, which leads to different rates. As shown in
Table~\ref{table:relation}, for the same price $p_{i}$, the
processing rate $\lambda_{p}$ of sort voting is lower than that
\textit{yes} or \textit{no} voting. In addition, the entire latency
depends on both how long a task is offered before it is accepted,
which depends on the reward offered, and also how long it takes to
complete the task, which depends on the task itself but not on
price. Unlike in the previous example, we now have to take such latency into
account as well.  Trying two budget allocations: evenly allocating
\$3 to two tasks; balancing budget according to difficulty, sorting
task with \$4 and filtering task with \$2, we get latencies as shown
in the figure. Once again, these are obviously not the only
allocations possible, and our interest is in finding the optimum,
compounded by the difficulty of predicting the uptake rate for any
reward level (notice that Table~\ref{table:relation} only gives us
values for a few price points), and of additionally folding in the
task completion time into the framework.

\end{motivation example}

As shown in the motivating example above, the allocation of budget
to the tasks matters to a great extent in terms of the overall HPU
processing latency. The difficulty of finding the optimal allocation
strategy is two-fold: 1) the latency of an atomic task is a random
variable which depends on the type of the task, the allocated
budget, and the current workforce market situation, therefore it is
non-trivial to predict the overall latency of a set of tasks,
particularly when they are of different types; 2) the search space
for finding an optimal solution is large so that efficient
algorithms and/or approximation tradeoffs are necessary.  The
promised payment has a minimum granularity(\$0.01 on AMT), which
renders the tuning process a discrete, rather than a continuous
optimization problem.

To address these challenges, the following contributions are made:

\begin{itemize}

\item In Section~\ref{section:2hpu}, we begin with HPU characteristics, develop a stochastic model to predict uptake rate as a function of reward amount, and show how to estimate the model parameters.  Using these results, we can determine the expected latency for any specific budget allocation choice.
\item In Section ~\ref{section:3tuning},we formally propose the \emph{H-Tuning Problem} to minimize the expected latency of a given set of tasks and propose probabilistic analysis and tuning strategies under three practical scenarios: \emph{Homogeneous, Repetition and Heterogeneous}.  In each case, we show how to solve an optimization problem with a large feasible space of possible budget allocation choices.
\item The performance of proposed strategies are verified on real crowdsourcing platform, and with simulation in Section ~\ref{section:4experiments}.
\end{itemize}

In addition, Section~\ref{section:5relatedwork} gives an overview of the related work. Section~\ref{section:6conclusion} makes conclusion to this work.

%%%%%%%%%%%%%%%%%%%%%Related Work%%%%%%%%%%%%%%%%%%%%%
\section{Related Work}\label{section:5relatedwork}

Leveraging the HPU in hope of better performance is an attractive
topic ever since the emergence of crowdsourcing applications. Many
recent works have studied various optimization issues associated
with the HPU\cite{ooi:vldb12:cdas,caleb:vldb12}. Most of them focus
on the quality issue in terms of answer
confidence\cite{ipei:hcomp10:quality,milo:icde12}, and some efforts
on the optimization of monetary cost\cite{faridani:hcomp11}.However,
in the effort of designing an industrial level computing module,
speed or latency is always one of the most significant concerns
among the various properties. Most of current works touch this issue
by reducing the number of queries issued to the
crowds\cite{guo:sigmod12,jason:vldb13,law:human11,aditya:sigmod12,
milo:icde13,milo:sigmod13,marcus:vldb11}.Whereas, considering the
HPU a new ``hardware'' for general human computation, a lower-level
clock-rate model, instead of the higher-level number of queries, is
far more entailed. Unfortunately, the stochastic human behavior
makes this model rather intractable. Several applications tried to
optimize the HPU's performance in real time in order to finish tasks
before a preset
deadline\cite{little:uist10:vizwiz,yan:mobisys10,galen:sci11:mobilization}.
But their approaches are highly application-dependent and thus hard
to adapt to a general framework; in addition, the ``deadline''
semantic does not support the batch processing scenario where a
general HPU usually meets.

Meanwhile, another practical methodology is developed by recruiting
a set of prepaid worker, so that they can wait online and process
the task immediately after publishing. In the work
of\cite{bernstein:uist11:realtime}, the authors propose such a
pre-paid model to instantiate a real-time respond crowdsourcing
interface, and a Retailer Model is adopted to describe the prepaid
workers behavior\cite{bernstein:ci12}. Following the Retainer Model,
one analytic effort on optimally organizing the microtasks can be
found in \cite{patrick:wec12:crowdmanager}.  Note that the prepaid
implementation differs greatly from this work: the tasks for prepaid
implementation entails high instantaneity, where the tasks are
expected to be finalized in several seconds(the payments are
relatively higher as well); however, the HPU tuning assumes a
system-level perspective, where the latency of the task set varies
over a larger range according to the specific requirement of the
database users. Last, the \textit{Queuing Theory} based model of
prepaid implementation cannot be tailored into the HPU scenarios
easily.

This work is most related to \cite{aditya:vldb15}, where the problem
of minimizing crowdsourcing latency is formulated into two
optimization issues: \emph{1. minimizing the completion cost of all
the tasks given deterministic deadline of every task}, and \emph{2.
minimizing the latency with constrained budget}. The objective of
the problem discussed in this work is virtually same with the second
issue in above work. However, our work is distinguished with
\cite{aditya:vldb15} in the following aspects. First, the latency of
a crowdsourcing task is modeled with two phases: the on-hold phase
and the processing phase. Such consideration is consistent with the
real world scenarios. However, \cite{aditya:vldb15} only considers
the latency of the tasks' acceptance. Secondly, the crowdsourcing
tasks can be processed both parallel (multiple tasks being processed
simultaneously) and sequentially (one task calls for multiple
answering repetitions, which are submitted one after another). Both
processing manners are studied in this work, while
\cite{aditya:vldb15} minimizes the latency with the implicit setting
of pure parallel processing.

%%%%%%%%%%%%%%%%%%%%%%%HPU%%%%%%%%%%%%%%%%%%%%%
\section{The HPU Model}\label{section:2hpu}

In this section, we begin with the basic crowdsourcing framework,
develop the HPU  model and demonstrate how to estimate HPUs'
parameters. In short, we lay the foundation for the optimization
problem we consider in the next section.

\begin{table} \label{table:relation}
\footnotesize \caption{HPU Processing Rate for Motivation Example}
\centering
\begin{tabular}{|c|c|c|}
  \hline
  \backslashbox{reward(\$)}{task type} & sorting vote & \textit{yes} or \textit{no} vote\\
  \hline
  2 & 2 & 3\\
  \hline
  3 & 3 & 5\\
  \hline
  1.5 & 1.5 & 2\\
  \hline
\end{tabular}
\end{table}
We begin with definitions of standard crowdsourcing concepts:
\begin{itemize}

\item[-]  \textbf{Requester}: A requester publishes tasks, collects answers, and makes the promised reward payments.  Database-wise, the \textit{requester} is the higher-level \textit{``executor''} as in \cite{marcus:vldb11} or \textit{``task manager''} in \cite{michael:sigmod11:crowddb}. The requester has also been called the``task-holder'', ``job-owner'' or ``builder''. %Please refer to Figure~\ref{fig:framework} for a system view demonstration.

\item[-] \textbf{Worker}: A worker (or crowd-worker) performs the actual human processing tasks.  A worker arrives at the market in a uniformly random manner, and she immediately chooses one of the tasks to work on. The preference of task selection is based on her utility maximization principle. After a period of time, the worker finalizes the task by returning the answer to the requester. Note that some research\cite{faridani:hcomp11} reports that the
worker activity on Amazon MTurk observes fluctuation along both a
daily and a weekly basis. However due to the scale of data-driven
micro-tasks, which are mainly light-weight voting, such long-term
fluctuation can be ignored, provided that we use parameters that
recurrent. In Section~\ref{section:running_para} we discuss the
practical methodology to infer the realtime system parameters.

%     Synonyms of the concepts includes ``crowds'', ``crowd worker'', ``oracle(in Active Learning)'' and so on.

\item[-] \textbf{Task}: A task is the most decomposed operation that a \textbf{worker} may
work on. Unfortunately, there are intrinsic limits of human
cognitive capacity\cite{law:human11}, and huge differences are
observed in the demographics of crowd workers\cite{ross:chi10}.
Consequently, to ensure the coherence and reliability of the human
answers, a \textbf{worker} is restricted to perform a set of most
basic operations like selecting from several options, ranking within
a couple of objects, connecting between figures, tagging images with
text and so on. Many of these human operations can be categorized
into voting, where a latent true option needs to be located with
some effort (a period of time).

In the literature, tasks have sometimes been called ``jobs'', ``HIT(Human Intelligent Tasks)'' and so on.  However, we reserve the word ``job'' for the following:
\item[-] \textbf{Job}: A job is what the requester is responsible for.  A job is accomplished by invoking tasks in parallel in one or more phases, with possible additional computation performed at the requester at the beginning and end of each phase.  In this paper, we will consider three different structures for tasks in phases, as we shall see below.

\end{itemize}

%With the specification of the fundamental elements, we put everything together to formulate the complete picture of a crowd-powered database system, whose brief illustration is found in Figure~\ref{fig:framework}. The end-desk user issues normal SQL statements to the system(sometimes with ``crowd'' keywords); the \textit{Query Parser} translates the statements into logical processing plan; the \textit{Executor \&
%Optimizer} identifies human-intrinsic operators(e.g. join by photo and audio by person) and optimizes the plan at a higher level; the \textit{Query Compiler} decomposes the human-intrinsic operators into HPU operations(e.g. voting); the \textit{Task Generator}
%produces HPU tasks for crowdsourcing platform; the \textit{HPU Tuner} groups the tasks and allocates budgets and then issues the tasks onto the HPU; in the end, the \textit{Answer Aggregator} collects and aggregates the human answers. The entire set of
%computational modules shown (other than the HPUs) are part of the Requester.

%\begin{figure}
%\centering
%\includegraphics[width=2.4in,height=1.5in]{framework}
%\caption{A Framework of Crowd-powered Database System}
%\label{fig:framework}
%\end{figure}

\subsection{Worker Selection Model}
Based on the definition of \textit{worker} above, in a workforce market, a worker appears and starts working on a task uniformly at any time. Meanwhile, the worker's preference among the candidate tasks relies on the subjective utility measurement.

\subsubsection{Worker Appearing Time}
The online workers enter the crowdsourcing market with a random manner. For a short period of time, like a few hours for platforms like Amazon MTurk, (according to the statistics of workers' arrival which is publicly released on AMT), the workers' arrival rate (the number of workers arrive within the unit time) can be regarded to be a constant number. Such property enables us to model the worker's appearing time with the following process. Denote the current workers' arrival rate with the constant number $\lambda$. For a time interval of fixed length $\Delta t$, the probability of \emph{No worker appears} equals to $(1-\Delta t \cdot \lambda)$. Suppose a task is submitted at time ``0'', and the task is accepted rightly after a worker arrives, the distribution of its acceptance can be derived as follows: $P(t_{acc} \leq s) = 1 - P(N(s)=0) = 1 - (1-\Delta t \cdot \lambda)^{\frac{s}{\Delta t}}$,,
%\begin{eqnarray}\begin{split}
%P(t_{acc} \leq s) = 1 - P(N(s)=0) = 1 - (1-\Delta t \cdot \lambda)^{\frac{s}{\Delta t}},
%%=&\lim_{\Delta \to 0}[(1-\Delta t \cdot \lambda)^{\frac{1}{\Delta t}}]^{t}\\
%%=&e^{-\lambda t}
%\end{split}\end{eqnarray}
where $t_{acc}$ is the time when the task is accepted and $N(s)$ denotes the number of arriving worker at time stamp $s$. Taking limit to $\Delta t$ gives the following expression: $P(t_{acc} \leq s) = 1 - \lim_{\Delta t \to 0}(1-\Delta t \cdot \lambda)^{\frac{s}{\Delta t}} = 1- e^{-\lambda s}$.
%\begin{eqnarray}\begin{split}
%P(t_{acc} \leq s) = 1 - \lim_{\Delta t \to 0}(1-\Delta t \cdot \lambda)^{\frac{s}{\Delta t}} = 1- e^{-\lambda s}.
%\end{split}\end{eqnarray}
Clearly, the acceptance time of a task follows exponential distribution on condition that the task is accepted once a worker arrives.

Recent research work in \cite{jwang:csdm11,faridani:hcomp11} delve
into more detailed analysis of when workers appear, with more
delicate consideration of time period and so on. Nevertheless, for an
encapsulated computation module, a major exponential model is
powerful enough for describing the latency characteristics.

\subsubsection{Task Preference}
In previous discussion, we make the assumption that a task is accepted once a worker arrives. However, workers have preferences over the tasks and tend to choose the task that can maximize her benefits. In other words, a task is accepted by an appearing worker with certain probability ``$p$''. Since we have pointed that after the submission of a task, the latency can only be adjusted though pricing, therefore $p$ is set to be variable affected by the task's price ``$c$'' (``$p(c)$''). Together with the worker's arrival rate, the probability of \emph{No worker accepts a task} is derived as: $(1- \lambda p(c)\Delta t)$ (when $p(c)=1$, such expression is equivalent to probability of ``No worker arrives'' presented in last part). Following the same procedure, the task's acceptance distribution is re-formulated as: $P(T<t) = 1 - e^{-\lambda_{c} t} = 1 - e^{-\lambda p\text{(c)} t}$,
%\begin{equation}
%P(T<t) = 1 - e^{-\lambda_{c} t} = 1 - e^{-\lambda p\text{(c)} t},
%\end{equation}
where $\lambda_{c}$ ($\lambda_{c} = \lambda p\text{(c)}$) is the joint acceptance rate of price $c$.
%Together with the model of the worker's appearing time, the distribution of task $t$ being accepted by a worker following the following relationship:

A detailed discussion of choice model can be found in \cite{faridani:hcomp11}. But to better estimate the latency behavior, in Section~\ref{section:running_para}, we present a real-time technique to infer parameters for tuning strategies.

\subsection{The HPU Latency}
Like in traditional CPU-based applications, when a
single task is published to the HPU, there will be two phases before
the answers are returned and collected: on-hold phase and processing
phase. The first one is the period from the task being published to
the task being chosen by a worker; the second one is the period
waiting for answer from the worker. Statistical research has been
conducted on several crowdsourcing platforms to capture the traits
of such latencies\cite{jwang:csdm11,faridani:hcomp11,mason:hcomp09}.

\begin{definition}[Latency]
The On-hold Latency $L_{o}$ of a task (or a batch of tasks) is the clock time from when the task is published to the time when it is accepted by a worker. The Processing Latency $L_{p}$ of a task (or a batch of tasks) is the clock time from when the task is accepted to the time when the answer is returned and collected by the system. The Overall Latency $L$ is the sum of $L_{o}$ and $L_{p}$: $L = L_{o} + L_{p}$.
%\begin{equation}
%L = L_{o} + L_{p}
%\end{equation}
\end{definition}
According to the worker appearing behavior proposed previously, we can derive that the distribution of the overall latency as follows. Let $\lambda_{o}$ and $\lambda_{p}$ denote the \emph{clock rates} of the process in On-hold and Processing phase respectively, and the probability density function of the latencies are as follows
\begin{eqnarray*}
\begin{split}
f_{o}(t)=\emph{pdf}(L_{o}\leq t)=\lambda_{o}\ue^{-\lambda_{o}t},
f_{p}(t)=\emph{pdf}(L_{p}\leq t)=\lambda_{p}\ue^{-\lambda_{p}t}
\end{split}
\end{eqnarray*}
Since the latency of On-hold phase depends on the attractiveness of a task towards the crowds, whereas the latency of Processing phase depends on the actual cognitive load of a task, we assume these two phases are independent from each other, which is supported by a recent study~\cite{yan:mobisys10}. Therefore, the probability density function for the overall latency $L$ can be derive as following.
\begin{eqnarray*}
\begin{split}
f_{L}(t)&=\emph{pdf}(L\leq t) =f_{o}(t)*f_{p}(t)=\int_{0}^{t}\lambda_{o}\ue^{-\lambda_{o}(t-u)}\lambda_{p}\ue^{-\lambda_{p}u}\,du \\
\displaystyle &=\frac{\lambda_{o}\lambda_{p}}{\lambda_{o}-\lambda_{p}}(\ue^{-\lambda_{p}t}-\ue^{-\lambda_{o}t})=\frac{\lambda_{o}\lambda_{p}}{\lambda_{o}-\lambda_{p}}(\ue^{-\lambda_{p}t}-\ue^{-\lambda_{o}t}),
\end{split}
\end{eqnarray*}
where ``$*$'' denotes the convolution operation of two pdf.

\subsubsection{Parallel Processing}
In order to complete tasks quickly, unrelated tasks will be
published simultaneously onto the crowdsourcing platforms.
Given a set of $k$ batch tasks, $B_k$, being processed by the HPU
simultaneously, the distribution for the overall latency of parallel processing is the maximum latency of all the tasks:
\begin{eqnarray*}\begin{split}
F_{para}(t)&=\emph{cdf}(L_{para}(B_{k})\leq t)=\prod_{i=1}^{k}\emph{cdf}(L(b_{i})\leq t)
\end{split}\end{eqnarray*} %\emph{cdf}(\max_{\forall b_{i}\in B_{k}}(L(b_{i}))\leq t)

%For certain situations, where all the tasks are of the same type,  the latencies of the second phase(Processing Stage) observe same coming rate $\lambda$. In such case, we maximum latency of the first phase(On-hold Stage) among all the tasks could be generalized as follows:
%\begin{eqnarray*}\begin{split}
%F_{PL(o)}(t)&=\emph{cdf}_{o}(\underset{\forall b_{i}}{\max}(L(b_{i}))\leq t)\\
%&=\prod_{i=1}^{k}\emph{cdf}_{o}(L(b_{i})\leq t)\\
%&=\prod_{i=1}^{k}[1-\sum_{i=1}^{k}\frac{e^{-\lambda_{o}^{i}t}}{k!}(\lambda_{o}^{i}t)^{k}]\\
%\end{split}\end{eqnarray*}

\begin{table}
\centering
\begin{footnotesize}
\caption{Summary of Notations}
\begin{tabular}{|c|c|}
  \hline
  Notation & Description \\
  \hline
  $t_{i}$ & an atomic task\\
  \hline
  $T_{N}$ & a set of atomic tasks with size $N$\\
  \hline
  $L(t_i)$ & latency of task $t_i$\\
  \hline
  $K$ & the maximum number of possible batches \\
  \hline
  $L_{o}^{b_{i}}$ & \emph{On-hold}  latency of batch $b_{i}$\\
  \hline
  $L_{p}^{b_{i}}$ & \emph{Processing} latency of batch $b_{i}$\\
  \hline
  $\lambda_{o}^{i}$ & the \emph{On-hold clock rate} of batch $b_{i}$\\
  \hline
  $\lambda_{p}^{i}$ & the \emph{Processing clock rate} of batch $b_{i}$\\
  \hline
  $B$ &  the total budget\\
  \hline
  $g_i$ & task group $i$, whose tasks are of $i$ repetitions\\
  \hline
  $p_{i}$ & payment for the task group $i$\\
  \hline
  $E(g_i)$ & the expected latency for the task group i\\
  \hline
\end{tabular}
\end{footnotesize}
\label{table:notation}
\end{table}

\begin{Example}
Revisiting the motivating examples of the introduction, we now have the machinery in place to discuss how we obtained the latencies shown in (Figure~\ref{fig:eg1} and (Figure~\ref{fig:eg2}.
The expectation of the longest task for the first example is
\begin{displaymath}
E[L]=\frac{1}{\lambda_{1}+\lambda_{2}}(1+2\frac{\lambda_{1}}{\lambda_{2}}+ \lambda_{2}(1+\frac{\lambda_{1}}{\lambda_{2}}+\frac{\lambda_{2}}{\lambda_{1}}))
\end{displaymath}
Based on Table~\ref{table:relation}, $E[case_{1}]=2.93(s)$ and
$E[case_{2}]=2.25(s)$, where the load-sensitive strategy is better.

A similar computation for the second example shows that the expected
latency becomes 3.5s and 2.7s respectively.
\end{Example}

\subsection{The HPU Running Parameters}\label{section:running_para}

The crowds workforce platform is always fluctuating, both in terms
of demographics and in population.  However, an exponential model
suffices as a good approximation. To support a robust tuning
strategy, we propose to statistically infer the parameters with
following two methodologies.

\subsubsection{Parameter Inference}

To infer the parameter $\lambda_o$, a ``probe'' program is
introduced, which publishes tasks with varying prices. The workers
who accept the task are simply required to make the submission as
soon as possible, so that the processing latency is small enough to
be neglected. Due to the specific application scenario, two
different inference methodologies could be adopted.

\textbf{Fixed Period}
The probe publishes sample tasks with the same type and price. After a fixed period $T_{0}$, the number of taken tasks as $N$ is observed.

\textbf{Random Period}
The probe publishes sample tasks with the same type and price at moment $t_{0}$. After $N$ tasks have been taken(or finished), track down the length of the period $T_{0}$ starting from $t_{0}$.

For both methodologies, under maximum likelihood estimate, the parameter $\lambda_o$ is given by $\hat{\lambda}_o=\frac{N}{T_{0}}$.
%\begin{displaymath}
%  \hat{\lambda}_o=\frac{N}{T_{0}}
%\end{displaymath}
Proof of the correctness of the inference can be found in Appendix Section~\ref{section:infer}. Further advanced sampling-based inference can be found in \cite{basawa:book80}.
The clock rate for the processing phase $\lambda_p$ is estimated with similar manner. This time, tasks of a specific type are published and the clock rate overall latency is estimated as: $\hat{\lambda}=\frac{N}{T_{0}}.$
%\begin{displaymath}
%  \hat{\lambda}=\frac{N}{T_{0}}.
%\end{displaymath}
Then $\lambda_p$ is estimated as: $\lambda - \lambda_o$, where $\lambda_o$ is the estimation of On-hold clock rate.
\subsubsection{Linearity Hypothesis}\label{section:hypo}
Without loss of generality, within a certain time interval, the price $c$ and the clock rate for the On-hold phase $\lambda_{o}(c)$ observes relationship with certain linearity. To provide better enhancement of the tuning strategy, we propose a Linearity Conjecture as following, which is the supporting property for strategy in Section~\ref{section:homo}. (The concrete values of the linearity between $c$ and $\lambda_{o}(c)$ does not affect the design of tuning strategy.)
\begin{hypothesis}[Linearity]
There exists constant values $k$ and $b$, such that the rate $\lambda_{c}$ and price $c$ follows $\lambda_{o}(c)=k\cdot c+b$.
\end{hypothesis}
The experiment part gives an empirically justifies this conjecture.

%%%%%%%%%%%%%%%%%%%%Tuning Strategies%%%%%%%%%%%%%%%%%
\section{Tuning Strategies}\label{section:3tuning}
In this section, the \emph{H-Tuning} problem is defined in the first place. Then the tuning strategies are developed according to three different scenarios.
%we first define and then solve the HPU tuning problem. In Subsection~\ref{section:problem}, we present the general formal definition of the expected-latency minimization problem. Then in the following subsections, we introduce three computation architectures with specific solutions.
\subsection{Problem Definition}\label{section:problem}
\begin{definition}[Latency Target]
A \emph{Latency Target} $L^{*}$ is a stochastic objective function for the tuning problem.
\end{definition}
Specific instantiation of $L^{*}$ will be presented in each scenario.

\begin{definition}[H-Tuning Problem]
Given a set of atomic tasks $T=\{t_{1},t_{2},\ldots,t_{N}\}$  with
size $N$, a discrete budget $B$, find an optimal budget allocation
strategy so that \emph{Latency Target} $L^{*}$ is minimized, without
exceeding the budget $B$.
\end{definition}

\subsection{Scenario \uppercase\expandafter{\romannumeral1} - Homogeneity}\label{section:homo}
\vspace{-0.5em}
\subsubsection{Scenario Description}

Scenario I is the most fundamental case. In this scenario, the
system is provided with a set of identical (in terms of difficulty) atomic tasks, which require the same number of running repetitions. 
All these atomic tasks is published simultaneously, and completed when all the tasks are solved for the required repetitions. 
A fixed budget is given at the very beginning
and the system is to come up with the budget allocation for each
atomic task before publishing them to the platform. The budget allocation is made to minimize the expected latency of all the
atomic tasks being solved.

\subsubsection{Tuning Strategy for Scenario \uppercase\expandafter{\romannumeral1}}

\nop{As is described in the previous section, Here, Now, we will
firstly illustrate the expected latency and then present the optimal
solution for budget allocation.}

The overall latency of all tasks being solved is equivalent to the maximum value of every single task's latency. 
Specifically, this is defined as
$L^{*}=L(T)=\max\left\{L(t_{i})|i=1,2,\ldots,N\right\}$. As is stated
earlier, the latency for each repetition is composed of two phases: the on-hold phase (Phase 1) and the
processing phase (Phase 2). The latency of both phases follows an
exponential distribution with parameters of $\lambda^{o}$ and
$\lambda^{p}$. 
The value of $\lambda^{o}$ is determined by the
allocated payment with a constant market condition, and the value
of $\lambda^{p}$ is determined simply by the nature. While our objective is to minimize the 
overall latency, the budget allocation does not affect the processing latency. 
Because of the identical nature of the processing time for all the tasks, the minimization of the on-hold 
latency leads to the minimum latency as well. Therefore, for Scenario \uppercase\expandafter{\romannumeral1}, the objective is changed to the 
minimization of the expected latency of the on-hold phase.
In the remaining part of this section, unless otherwise specified,
we use the term ``expected latency'' referring to the expected
latency in Phase 1. Before giving the optimal solution of the budget
allocation problem for Scenario
\uppercase\expandafter{\romannumeral1}, we introduce the following
Lemmas and Theorems.
\vspace{-0.5em}
\begin{lemma}\label{LEMMA1}
Given two identical atomic tasks $t_{1}$ and $t_{2}$, both requiring
to be run one round, a fixed budget of $B$ unit payment, allocating
both $t_{1}$ and $t_{2}$ with $\frac{B}{2}$ (or if $B$ is odd,
allocating these two atomic tasks with $\left \lfloor \frac{B}{2}
\right \rfloor$ and $\left \lfloor \frac{B}{2} \right \rfloor+1)$
unit payments leads to the minimum expected latency of completing
$t_{1}$ and  $t_{2}$.
\end{lemma}
\vspace{-0.5em}
\begin{proof} Please refer to Appendix Section~\ref{section:proof_Lemma1}
\end{proof}

Then, Lemma $2$  shows that for one atomic task with multiple
repetitions, allocating budget evenly to each repetition
will minimize the expected latency.
\vspace{-0.5em}
\begin{lemma}\label{LEMMA2}
For atomic task $t$ which needs to be run $m$ repetitions, and
a fixed budget of $B$ unit payment, allocating $B$ evenly to each
repetition of $t$ leads to the minimum expected latency.
\end{lemma}
\vspace{-0.5em}
\begin{proof}  Please refer to Appendix Section~\ref{section:proof_Lemma2}
\end{proof}

With the above two lemmas, we have Theorem $1$, which produces the budget allocation plan to 
minimize the expected latency.

\begin{theorem}\label{THEOREM1}

Given two identical atomic tasks which require to be run for the
same number of times and a fixed budget of $B$, allocating the budget evenly
to each repetition of all the atomic tasks leads to the minimum
expected latency.
\end{theorem}

\begin{proof}  Please refer to Appendix
Section~\ref{section:proof_even_allocation}.
\end{proof}

\label{THEOREM1} directly leads to the optimal budget plan, whose operations are shown in Algorithm
\ref{alg:EA}. As the optimal solution is obtained analytically, EA is conducted with $O(1)$ 
time complexity.

\begin{algorithm}[!ht]
  \small{
  \caption{\small{Even Allocation (EA)}}
  \label{alg:EA}
  \KwIn{budget $B$, atomic task set $T={t_{1},t_{2},\ldots,t_{N}}$, $m$ required repetition rounds}
  \KwOut{allocation of payment $P={p_{1}, p_{2},\ldots, p_{N}}$}
  %\SetVline
  \eIf{$B\leq m*N$}{
              \Return the budget is not enough\;
            }{
              $\delta=\left \lfloor B/mN\right \rfloor$ and each repetition of all the atomic tasks is allocated with $\delta$ unit payment\;
              $\gamma=\left \lfloor(B ~mod~ mN)/N\right \rfloor$ and select $\gamma$ repetitions from each atomic task randomly.  Increase the payment for the selected rrepetitions by one unit\;
              $\sigma=(B ~mod~ mN)~mod~ N$ and select $\sigma$ repetitions from $\sigma$ random atomic tasks whose payment is not increased in the previous step. Increase the payment of the selected repetition rounds by one unit\;
              }
              }
\end{algorithm}

\vspace{-0.5em}
\subsection{Scenario \uppercase\expandafter{\romannumeral2} - Repetition}\label{section:rep}
In this section, we take one more step forward: despite the identical difficulty, the tasks require different running repetitions.

\subsubsection{Getting the Expected Latency}
As the tasks require different number of running repetitions, the closed form of overall latency's pdf will become intractable when tasks come with large quantity. Thus, it's impossible to get the deterministic optimal solution. To address this challenge, the overall latency is processed approximately, based on which the optimal budget plan is derived.
Specifically, tasks are grouped according to the running repetitions. Then the overall latency is approximated with the sum of latency of all the task groups.

\textbf{Group of Single Round} Given task group $g$ which is composed of atomic tasks $t_{1}$,$t_{2}$,$\ldots$,$t_{n}$, requiring to be run for single round. According to the definition, the latency of $g$, which is denoted by $L(g)$, equals to $\max{(L(t_{1}),L(t_{2}),\ldots,L(t_{n}))}$. Let \\ $x_{1}=\min{\{(L(t_{1}),L(t_{2}),\ldots,L(t_{n}))\}}$,
%\begin{displaymath}
%x_{1}=\min{\{(L(t_{1}),L(t_{2}),\ldots,L(t_{n}))\}}
%\end{displaymath}
which means the first completion of all the tasks within the group, and then let \\
$x_{2}=\min{\{\{(L(t_{1}),L(t_{2}),\ldots,L(t_{n}))\}-x_{1}\}}$
%\begin{displaymath}x_{2}=\min{\{\{(L(t_{1}),L(t_{2}),\ldots,L(t_{n}))\}-x_{1}\}}
%\end{displaymath}
 which means the second completion of all the atomic tasks within the group, and the like, $x_{n}=\min{L(g)-\sum_{i=1}^{n-1}x_{i}}$, which means the last completion of all the atomic tasks.
It can be derived that $L(g)=\sum_{i=1}^{i=n}x_{i}$.
As $x_{i}\sim exp(\lambda \ast i)$, $L(g)$ can be regarded as the sum of $n$ exponential variables. Therefore, $L(g)=\sum_{i=1}^{n}\frac{1}{\lambda \ast i}$.
%\begin{displaymath}L(g)=\sum_{i=1}^{n}\frac{1}{\lambda \ast i}\end{displaymath}

\textbf{Group Multiple Rounds} Before we turn to the study of the
probabilistic model of the task group of multiple repetition rounds,
the following lemma is needed to show the probabilistic property of
the task which requires multiple running repetitions.

\begin{lemma}
Let $t$ denote an atomic task which needs to be run for $k$
repetition rounds, the latency of $t$ follows Erlang distribution of
parameter $k$ and $\lambda$, which is $L(t)\sim Erl\left
\{k,\lambda\right \}$ \label{LEMMA3}
\end{lemma}
\begin{proof} Please refer to Appendix Section~\ref{section:proof_Lemma3}.
\end{proof}
\vspace{-0.5em}
Now we can get the expected latency of the task group through the
following deduction. Suppose we are given a task group $g$, which is
composed of a set of tasks $\{t_{1}, t_{2},\ldots, t_{n}\}$
and each task is needed to be run for $k$ repetition rounds.
Let $L\{g\}$ denote the latency of the task group $g$, then we can
have the following relationship: $L\{g\}=L\{max{\{x_{1}, x_{2},\ldots, x_{n}\}}\}$.
%\begin{displaymath}L\{g\}=L\{max{\{x_{1}, x_{2},\ldots, x_{n}\}}\}\end{displaymath}
Let $F(t)$ denote the cumulative distribution function (cdf) and $f(t)$ denote the probability density function (pdf) of the latency of the atomic tasks respectively. Let $F_{g}(t)$ denote the cumulative distribution function (cdf) and $f_{g}(t)$ denote the probability density function (pdf) of the latency of the task group respectively. The following relationship can be derived:
\begin{eqnarray}
\nonumber &F_{g}(t)=F^{n}(t), \text{ }
\nonumber &f_{g}(t)=n*F^{n-1}(t)*f(t)
\end{eqnarray}
With the above relationship, we get get the expression of the expected latency of the task group as:
\begin{displaymath}
E\{L(g)\}=\int_{0}^{\infty}f_{g}(t)*tdt=\int_{0}^{\infty}n*F^{n-1}(t)*f(t)*tdt
\end{displaymath}
According to the conclusion of $lemma$ 3, the latency of the tasks within the task group follow Erlang distribution $Erl(k, \lambda)$, thus the expected latency of the task group is derived as:
\begin{displaymath}
E\{L(g)\}=\int_{0}^{\infty}n*F_{E}^{n-1}(k,\lambda,t)*f_{E}(k,\lambda,t)*tdt
\end{displaymath}
among which $F_{E}(k,\lambda,t)$ and $f_{E}(k,\lambda,t)$ denote the $cdf$ and $pdf$ of the Erlang distribution $Erl(k, \lambda)$.

\textbf{Approximate Expected Latency}
As is stated in the previously, the close form of expected latency is intractable when the number of tasks is huge. Thus, we use the sum of the expected latency of all the task groups to approximate the real function. There are two reasons for such approximation: one is that the sum of the expected latency of all the task group lays the upper bound of the expected latency of all the atomic tasks; the other is that the expected latency of all the atomic tasks will decrease while the sum of the expected latency of the task group is going down.

\subsubsection{Tuning Strategy for Scenario \uppercase\expandafter{\romannumeral2}}
%In order to minimize the expected latency of all the atomic task, the following optimizing problem is defined.

%\textbf{Optimal Budget Allocation}
%A dynamic programming based algorithm is developed here to get the optimal budget plan. 
Let $E_{g_{1}}, E_{g_{2}},\ldots, E_{g_{n}}$ denote the expected latency of the task group $g_{1}, g_{2},\ldots, $ $g_{n}$. Let $b_{g_{1}}, b_{g_{2}},\ldots, $ $b_{g_{n}}$ denote the allocated payment of task group $g_{1}, g_{2},\ldots,$ $g_{n}$. The optimizing problem is defined as
$\min \sum_{i=1}^{n}E_{g_{i}}$
$s.t. \sum_{i=1}^{n}b_{g_{i}}\leqslant B$.

A dynamic algorithm is designed as follows to solve such minimization problem. The outer loop of the algorithm increases the task payment from 1 to $B'$ ($B'=B-\sum_{i=1}^{n}u_{i}$). Within each loop, it takes $O(n)$ operations to find the optimal payment given the current budget. Apparently, the overall time complexity for algorithm 2 turns out to be $O(nB')$

\begin{algorithm}[!ht]
  \footnotesize{
  \caption{\small{Repetition Algorithm (RA)}}
  \label{alg:RA}
  \KwIn{budget $B$, task group $G={g_{1}, g_{2},\ldots, g_{n}}$}
  \KwOut{allocation of payment $P={p_{1}, p_{2},\ldots, p_{n}}$}

  \For {$i=1$ to $n$}
    {$p_{i}(0)=1$}
     $B'=B-\sum_{i=1}^{n}u_{i}$;\\
     $E_{0}(0)=\sum_{i=1}^{n}E_{i}(P_{i}(0))$;\\
  \For {$x=1$ to $B'$}
  {
  $E_{0}(x)=\min \{E_{0}(x-1),$ 
  $\min{ \{E_{0}(x-u_{i})-[E_{i}(p_{i})-E_{i}(p_{i}+1)]|u_{i}\leq x\} }  \} $\;
  \If {$E_{0}(x-1) \leq \{E_{0}(x-u_{i})-[E_{i}(p_{i})-E_{i}(p_{i}+1)]|u_{i}\leq x\}$ }
{$\theta=\underset{i}{\operatorname{argmin}} \{E_{0}(x-u_{i})-[E_{i}(p_{i})-E_{i}(p_{i}+1)]|u_{i}\leq x\}$\;
  $p_{\theta}(x)=p_{\theta}(x-1)+1$;}

  }

  }
  $\forall i=1,\ldots,n, p_{i}=p_{i}(b)$
\end{algorithm}

\vspace{-1em}
\subsection{Scenario \uppercase\expandafter{\romannumeral3} - Heterogeneous}\label{section:heter}

In Scenario \uppercase\expandafter{\romannumeral3}, the tasks are heterogeneous (in terms of difficulty) and need to be run for different numbers of repetitions. While dealing with the latency of first two scenarios, we only take the latency of Phase $1$ into account, whose reasons are two fold: the first one is that the payment does not change the latency of Phase 2, the second one is that the latency of Phase $2$ is identical for all the atomic tasks since all the tasks are homogeneous in terms of task nature. However, these properties no longer hold in Scenario \uppercase\expandafter{\romannumeral3} as different tasks require different processing time. 

For this scenario, some tasks are easier to solve, which produce smaller processing latency, while others are harder to solve, which lead to longer processing latency. As a result of such character, the previous tunning strategies do not apply well to the current problem, as the tuning result may be jeopardized by the tasks whose processing latency is significantly larger than others'. One extreme situation is that the latency of Phase 2 of some atomic tasks is so long that the overall latency of completing all the atomic tasks will be approximately equal to the expected latency such atomic task. We call such kind of atomic tasks as ``most difficult task''. It is obvious that such type of atomic tasks generate stronger influence to the overall latency than the others.

%This effect can be clearly illustrated by the following example. The first picture shows the result when the latency of Phase 1 is minimized. The second picture shows the result when the latency of Phase 2. From this picture, we can see that the latency of Phase 2 of some atomic task is much larger than that of others', thus exerting great delay effect on the final result. One extreme case is that when the latency of some atomic task is significantly larger than that of others', the expected latency of completing all the atomic tasks will be equal to the expected latency of the longest task. The third picture shows the result of 'borrowing' extra payment from other atomic tasks to cut down the largest latency. We can see that the overall latency is further decreased compared with picture 2 where only the latency of Phase 1 is minimized.

In order to relieve the delaying effect caused by the ``most difficult tasks'', we make the following adaption to the tuning strategy. For previous tuning strategies, only the latency of Phase 1 is considered. While in Scenario \uppercase\expandafter{\romannumeral3}, two objectives will be minimized simultaneously: one objective is still the latency Phase 1, the other one is the latency of the ``most difficult task'', which is equivalent to the largest expected latency of all the atomic tasks. The reason of introducing the first objective is the same as previous scenarios, which is the allocation of payment only changes the latency of Phase 1, while the reason of introducing the second objective is to confine the delay effect caused by the ``most difficult task''. Here the second objective serves as the penalty function to avoid the appearance of the situation where the latency of some atomic tasks is significantly longer than that of others'. One more point needs to be clarified is that we can't simply minimize the second objective because the minimization of the latency of the ``most difficult work'' doesn't necessarily lead to minimum latency of completing all the atomic tasks.

Formally, the objective function is defined as follow.
Let $G=\{g_{1},g_{2},\ldots,g_{n}\}$ denote the task group (the grouping operation is performed to all the atomic tasks so that the tasks of identical type and repetition fall into the same group, which is slightly different from Scenario \uppercase\expandafter{\romannumeral2}). Let $L^{1}(g_{i})$ and $L^{2}(g_{i})$ denote the latency of Phase $1$ and Phase $2$ of $g_{i}$ respectively. Objective $1$ is the expected latency of Phase $1$ of all the atomic tasks, which is denoted by $O_{1}$ and $O_{1}=E\{L^{1}(G)\}$. Objective $2$ is the sum of the expected latency of Phase $1$ and Phase $2$ of the most difficult atomic tasks, which is denoted by $O_{2}$ and $O_{2}=$ \\
$max{ \{  E\{L^{1}(g_{i})\} + E\{L^{2}(g_{i})\}|i=1,\ldots,n  \}   }$. Given the budget of $B$ unit payment and let $p_{i}$ denote the payment allocated to group $g_{i}$, the optimizing problem is defined as:
$\min \{O_{1},O_{2}\}$ $s.t.$\\
$ \sum_{i=1}^{n}p_{i}\leq B$.
Here, we adopt a ``Compromise strategy'' to solve the above two objective optimization problems. Firstly, the ``Utopia Point'' ($UP$) is calculated, which refers to the point where both objectives are optimized independently under the given constraints. In the second place, the ``Closeness'' ($CL$) is defined as the first order distance between the objective point $OP$ and $UP$. The ``Closeness'' is minimized under the given constrains, and the corresponding solution will serve as the optimal solution.
The definition of ``$UP$'', ``$OP$'', and ``$CL$'' are formally presented as follows.
\vspace{-0.5em}
\begin{definition}[Utopia Point]
Let $O_{1}^{*}= \{ \min O_{1}:s.t. \\ \sum_{i=1}^{n}p_{i}\leq B$ \} and $O_{2}^{*}= \{ \min O_{2}:s.t.\sum_{i=1}^{n}p_{i}\leq B \}$. The Utopia Point is defined as $UP=(O_{1}^{*},O_{2}^{*})$.
\end{definition}
\vspace{-1.0em}
\begin{definition}[Objective Point]
Let $O_{1}$ and $O_{2}$ denote the objective value of the current payment allocated to each task group. The Objective Point is the two dimensional position determined by \{$O_{1}$,$O_{2}\}$.
\end{definition}
\vspace{-1.5em}
\begin{definition}[Closeness]
The Closeness equals to the first order distance between $UP$ and $OP$: $CL=\left \|OP-UP \right \|$.
\end{definition}

Here, the optimal budget plan is equivalent to the minimization of the following problem:
$\min {CL}: s.t.\sum_{i=1}^{n}p_{i}\leq B$. Such a problem can be optimally solved with dynamic programming, whose procedures are shown with Algorithm 3. Similar with algorithm 2, the dynamic programming runs with $O(nB')$ iterations to achieve the optimal solution.

\begin{algorithm}[!ht]
  \small{
  \caption{\small{Heterogeneous Algorithm (HA)}}
  \label{alg:HA}
  \KwIn{budget $B$, task group $G={g_{1}, g_{2},\ldots, g_{n}}$}
  \KwOut{allocation of payment $P={p_{1}, p_{2},\ldots, p_{n}}$}

  \For {$i=1$ to $n$}
    {$p_{i}=1$}
     $B'=B-\sum_{i=1}^{n}u_{i}$;\\
     $CL_{0}=\left \|OP_{0}-UP \right \|$;\\
  \For {$x=1$ to $B'$}
  {
  $CL_{x}=\min \{CL_{x-1},
  \min \{CL_{x-u_{i}}(++p_{i}(x-u_{i})) $\\
  $ |u_{i}\leq x,i=1,\ldots, n\}  \} $;\\

  \If {$CL_{x-1}  $\\$ \leq\{CL_{x-u_{i}}(++p_{i}(x-u_{i}))|u_{i}\leq x,i=1,\ldots, n\}$ }
{$\theta=\underset{i}{\operatorname{argmin}} \{CL_{x-u_{i}}(++p_{i}(x-u_{i}))|u_{i}\leq x,i=1,\ldots, n\}$;\\
  $p_{\theta}(x)=p_{\theta}(x-1)+1$;}

  }
  }
  $\forall i=1,\ldots,n, p_{i}=p_{i}(b)$
\end{algorithm}

\begin{figure*}[htbp] \centering
\subfigure[Homogeneous($\lambda=1+p$)] { \label{fig:1}
\includegraphics[height = 1.0in,width=0.63\columnwidth]{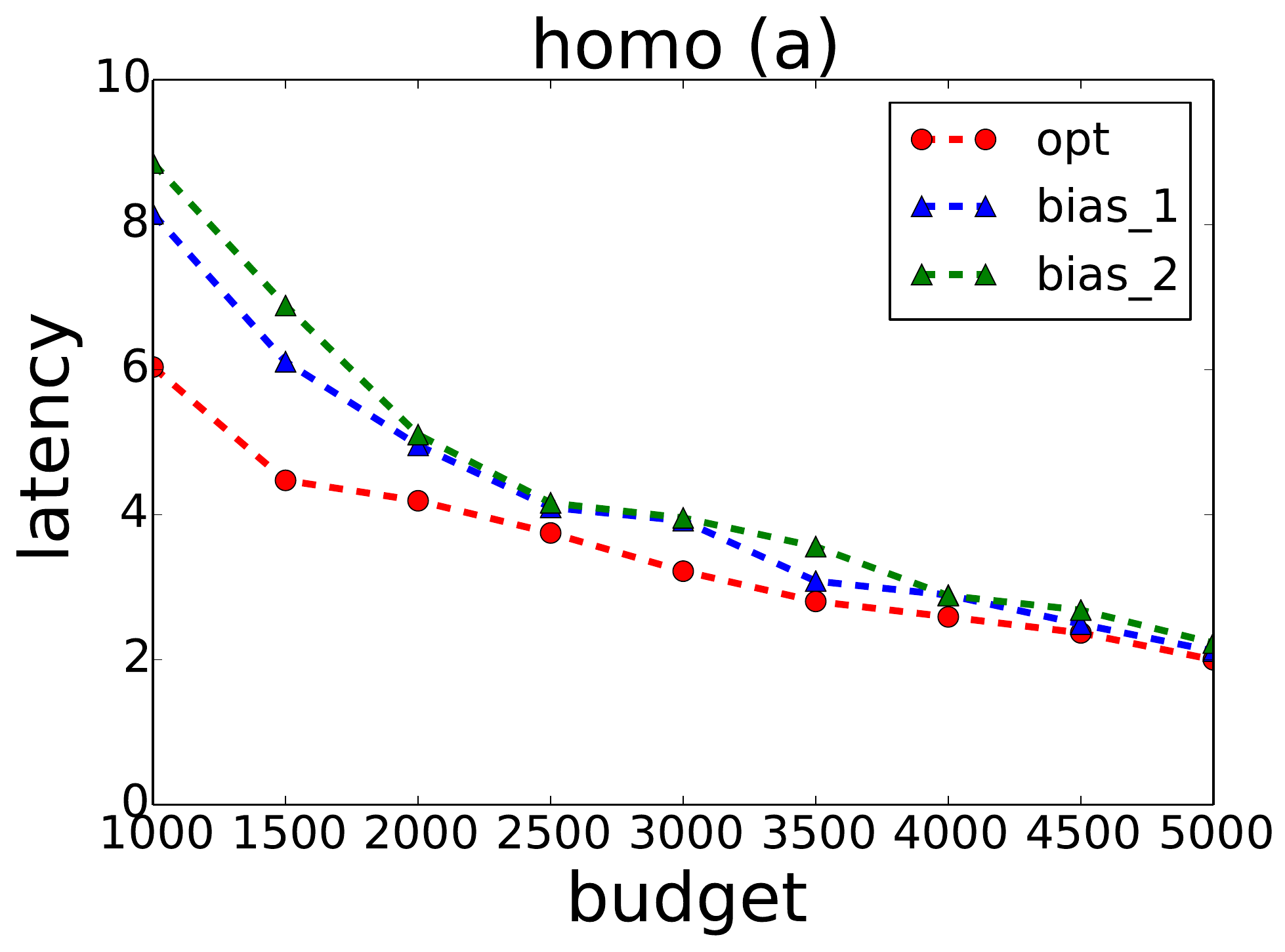}
}
\subfigure[Homogeneous($\lambda=10p+1$)] { \label{fig:2}
\includegraphics[height = 1.0in,width=0.63\columnwidth]{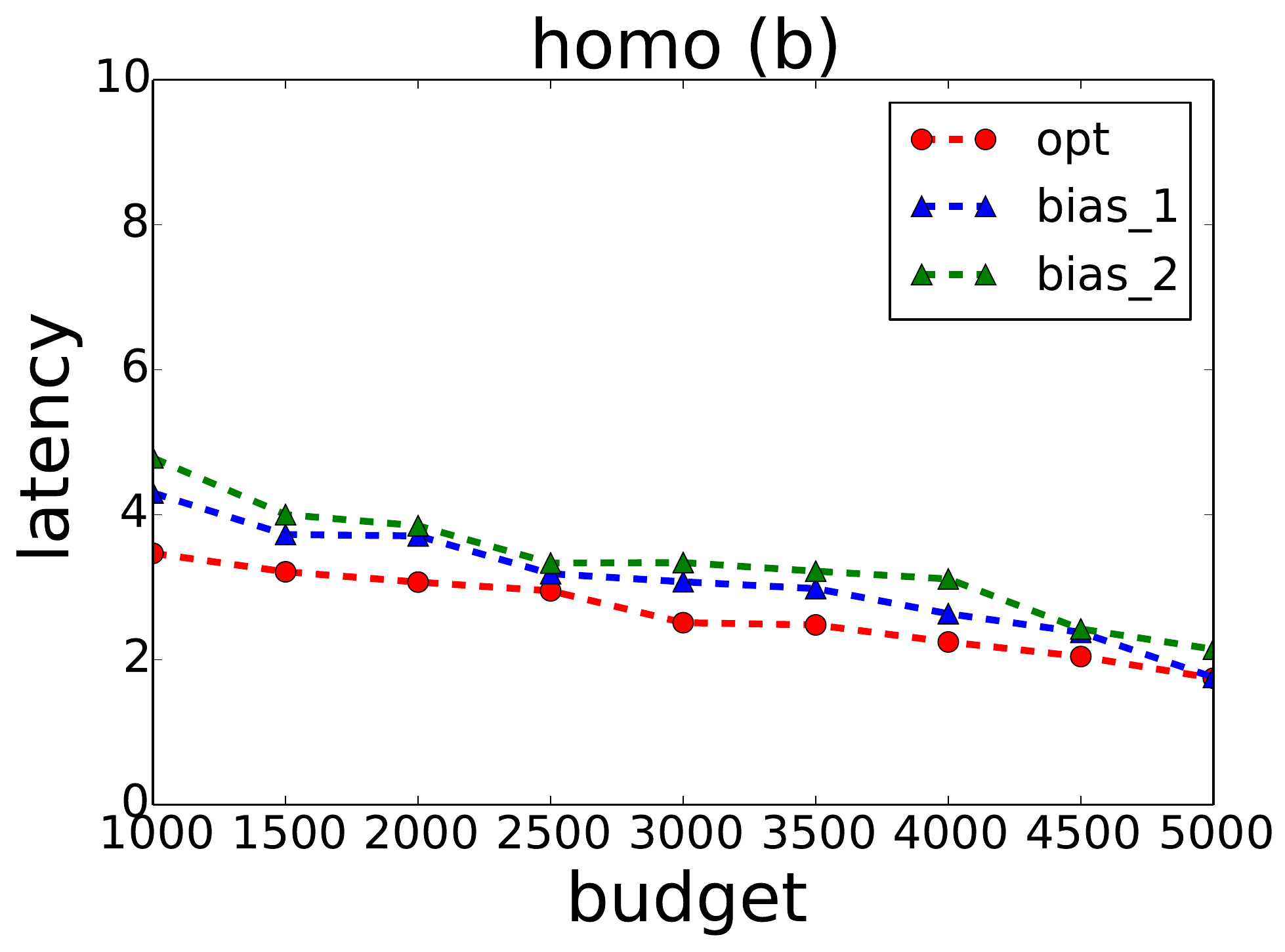}
}
\subfigure[Homogeneous($\lambda=0.1p+10$)] { \label{fig:3}
\includegraphics[height = 1.0in,width=0.63\columnwidth]{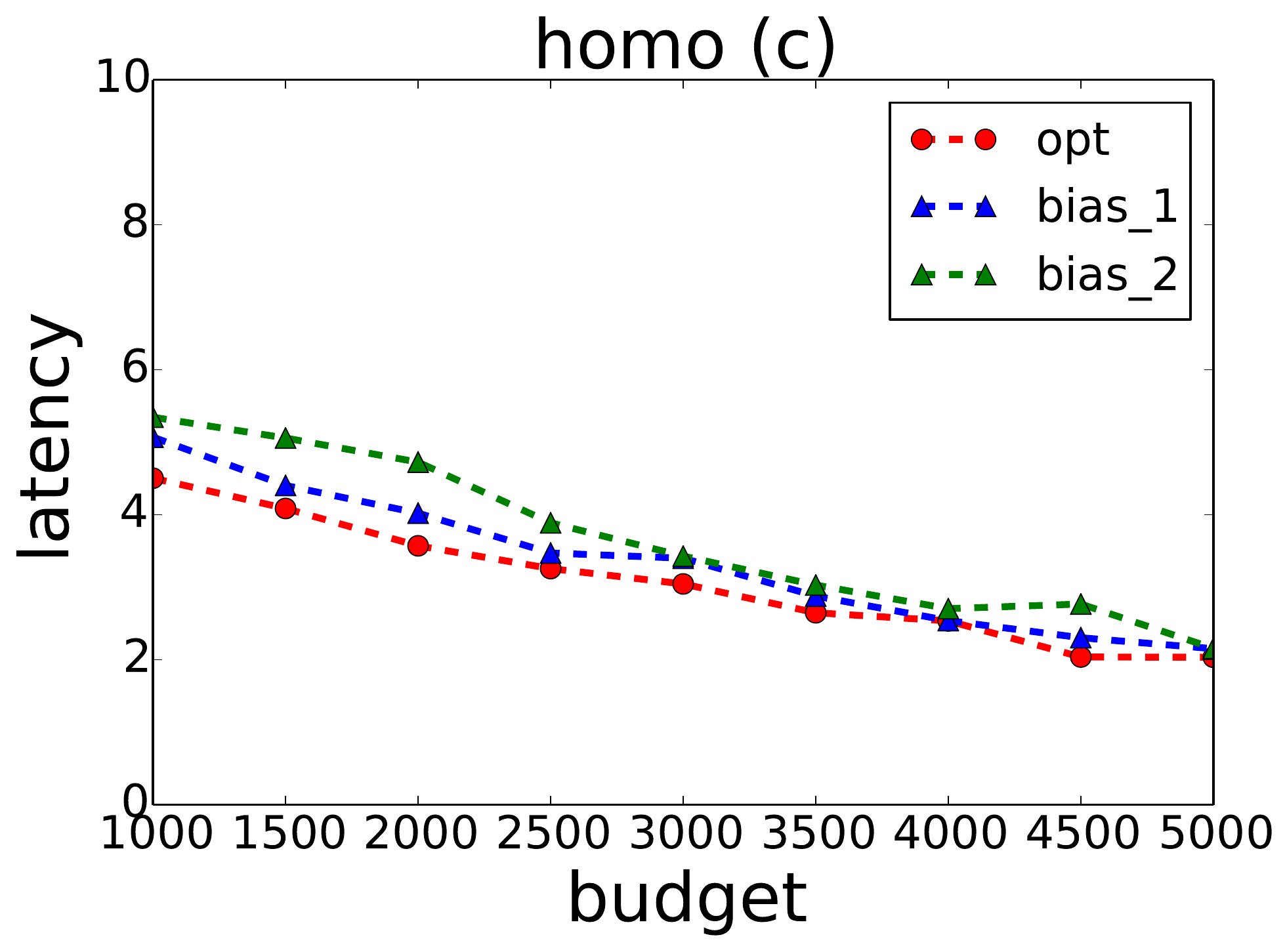}
}
\subfigure[Homogeneous($\lambda=3p+3$)] { \label{fig:4}
\includegraphics[height = 1.0in,width=0.63\columnwidth]{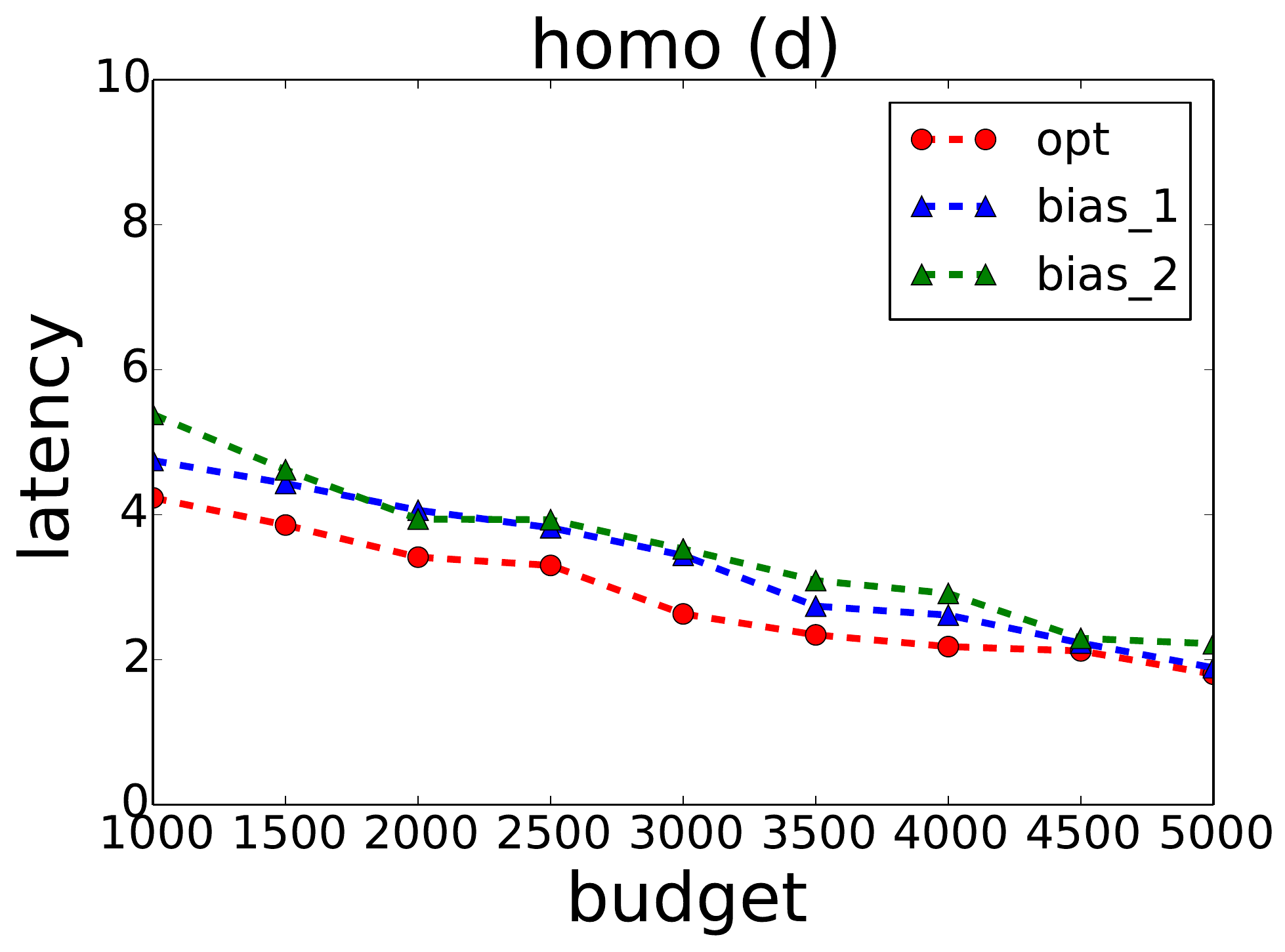}
}
\subfigure[Homogeneous($\lambda=1+p^{2}$)] { \label{fig:5}
\includegraphics[height = 1.0in,width=0.63\columnwidth]{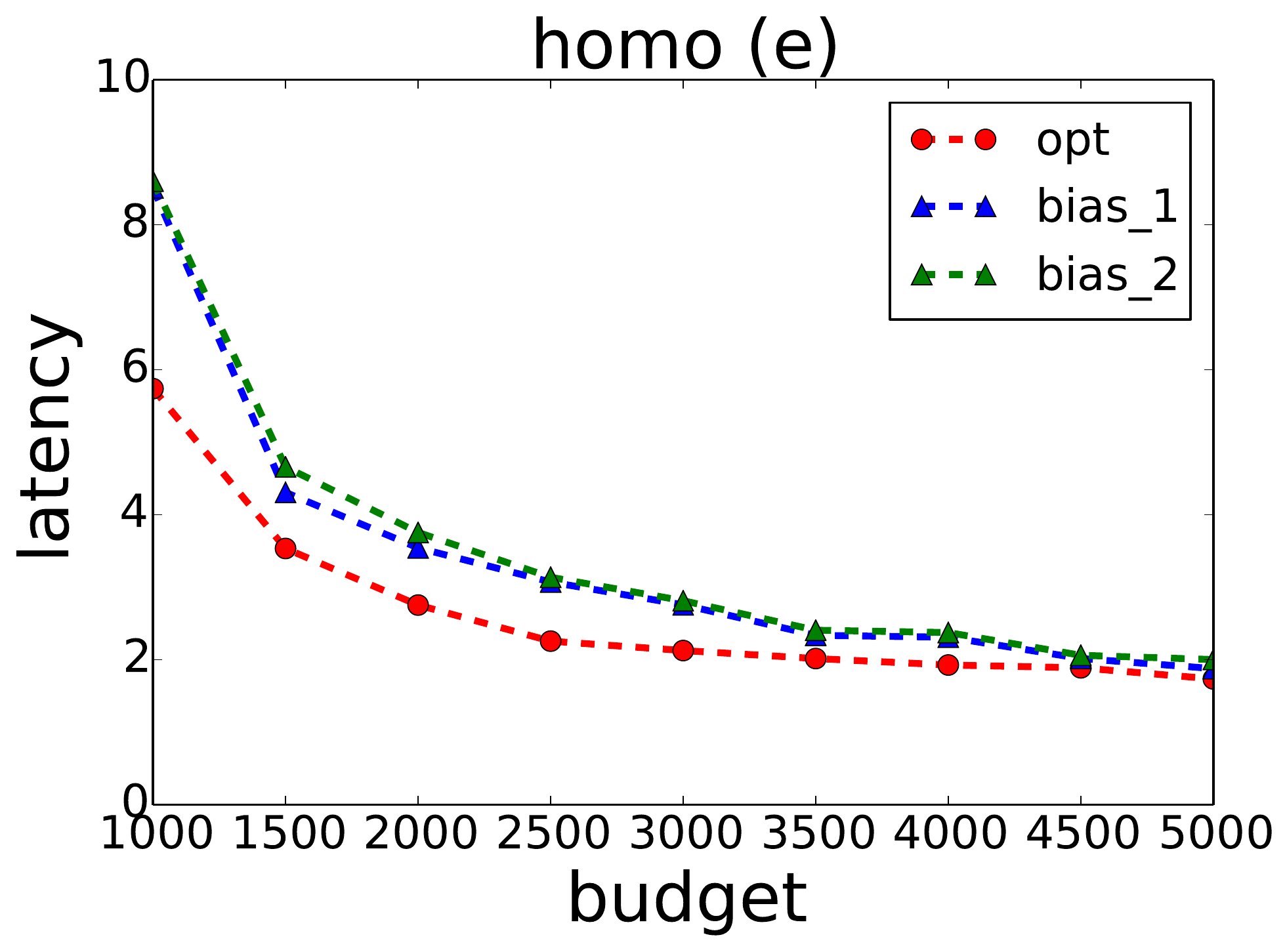}
}
\subfigure[Homogeneous($\lambda=\log{1+p}$)] { \label{fig:6}
\includegraphics[height = 1.0in,width=0.63\columnwidth]{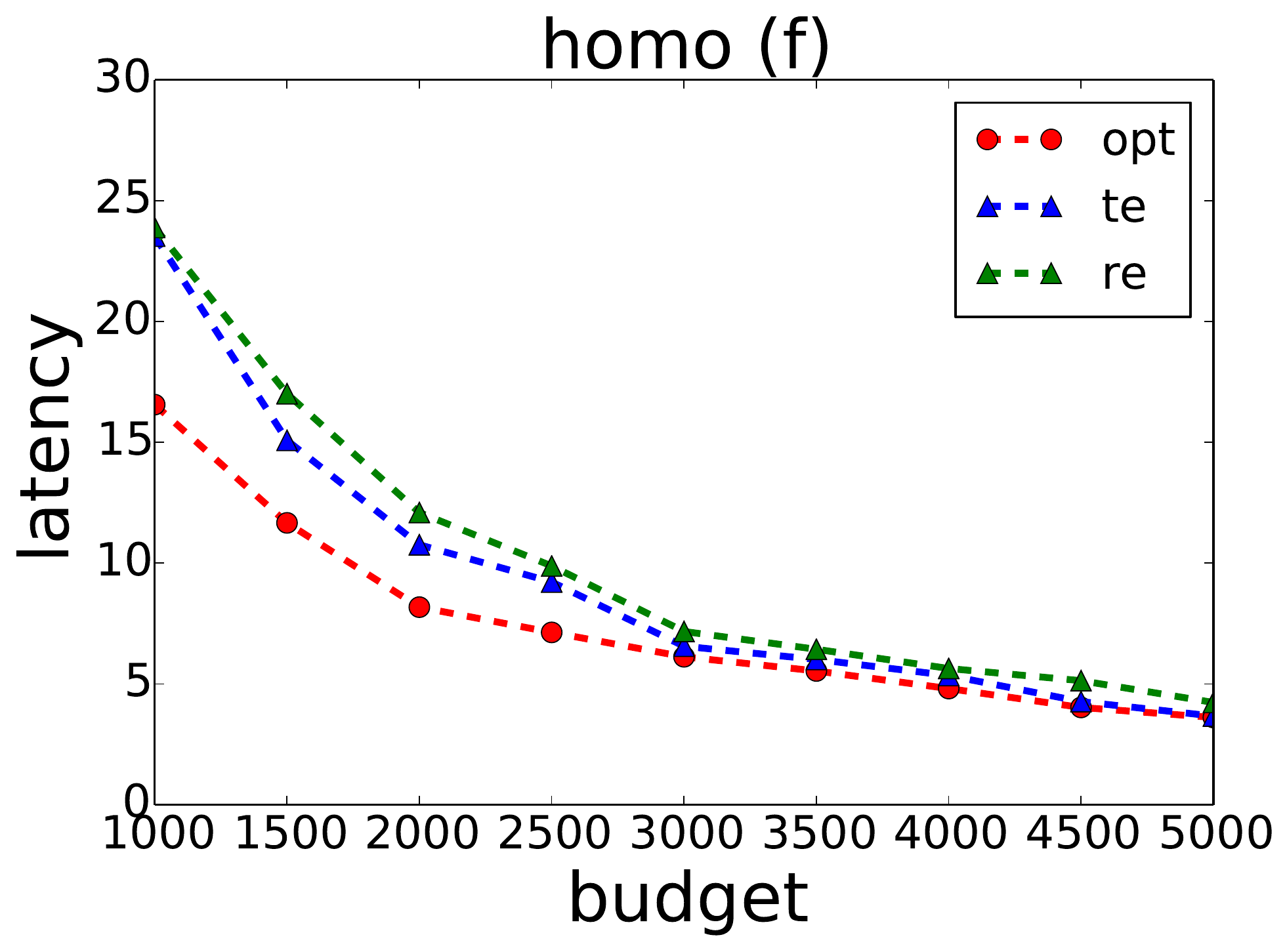}
}
\subfigure[Repetition($\lambda=1+p$)] { \label{fig:7}
\includegraphics[height = 1.0in,width=0.63\columnwidth]{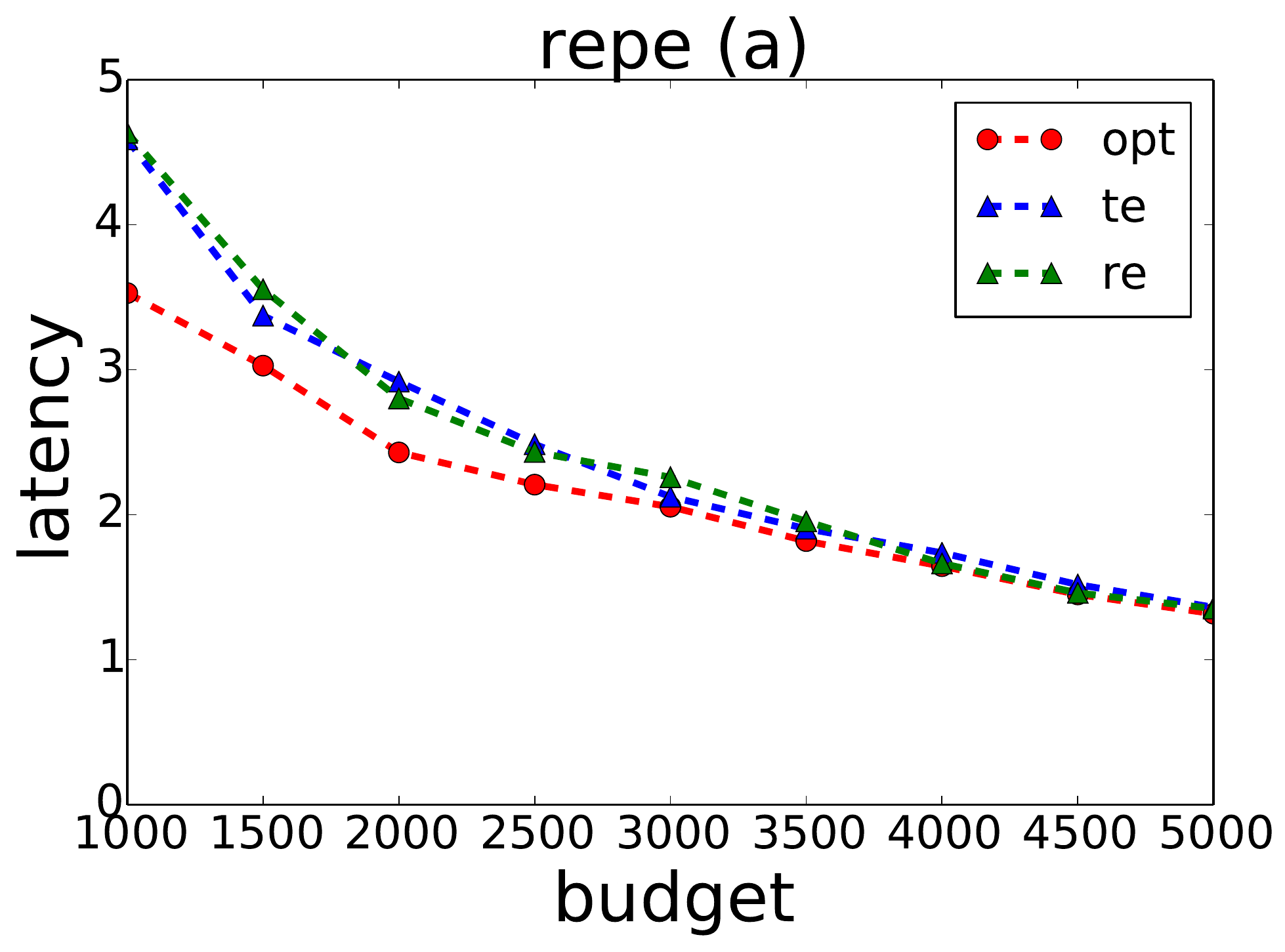}
}
\subfigure[Repetition($\lambda=10p+1$)] { \label{fig:8}
\includegraphics[height = 1.0in,width=0.63\columnwidth]{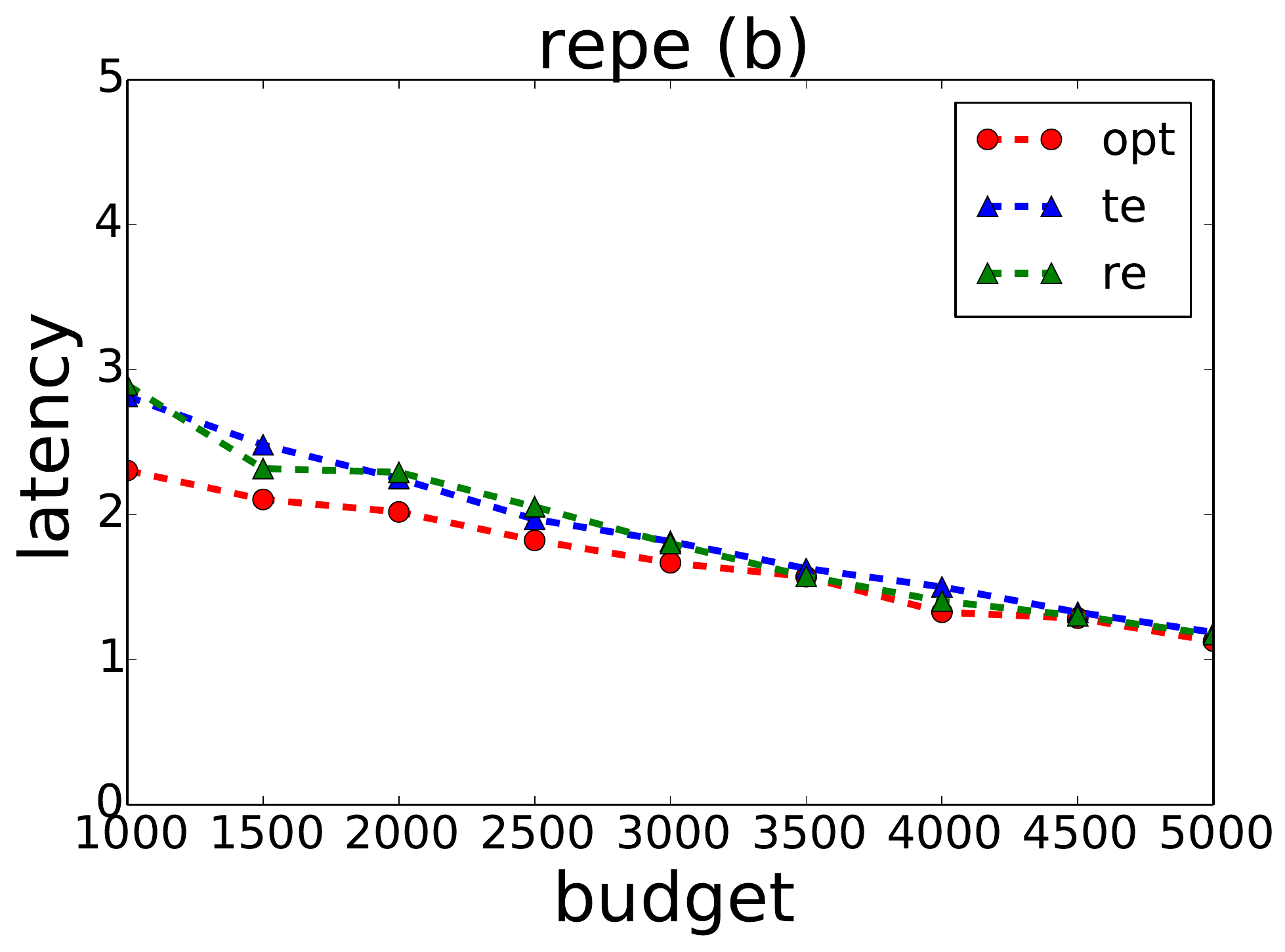}
}
\subfigure[Repetition($\lambda=0.1p+10$)] { \label{fig:9}
\includegraphics[height = 1.0in,width=0.63\columnwidth]{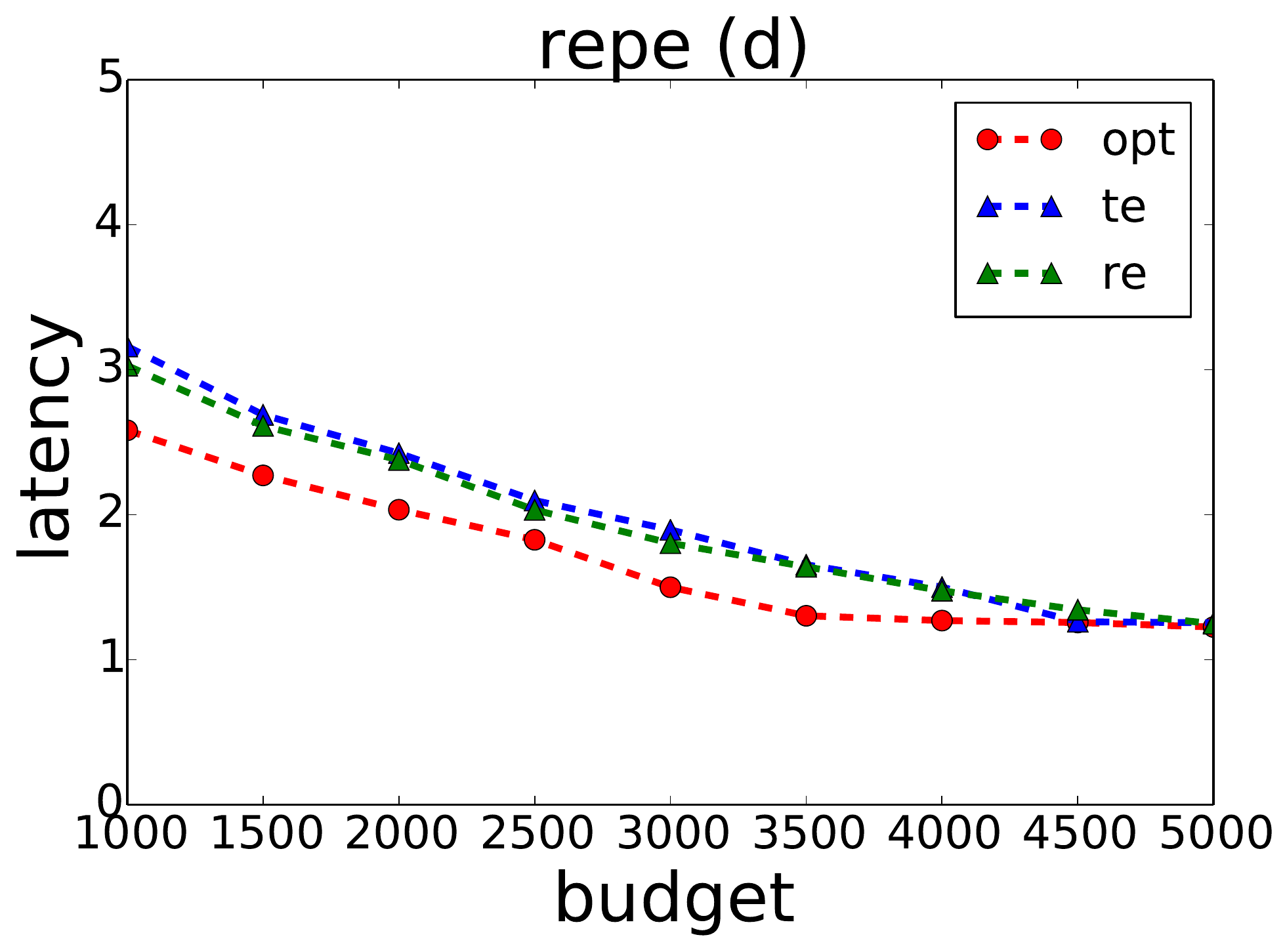}
}
\subfigure[Repetition($\lambda=3p+3$)] { \label{fig:10}
\includegraphics[height = 1.0in,width=0.63\columnwidth]{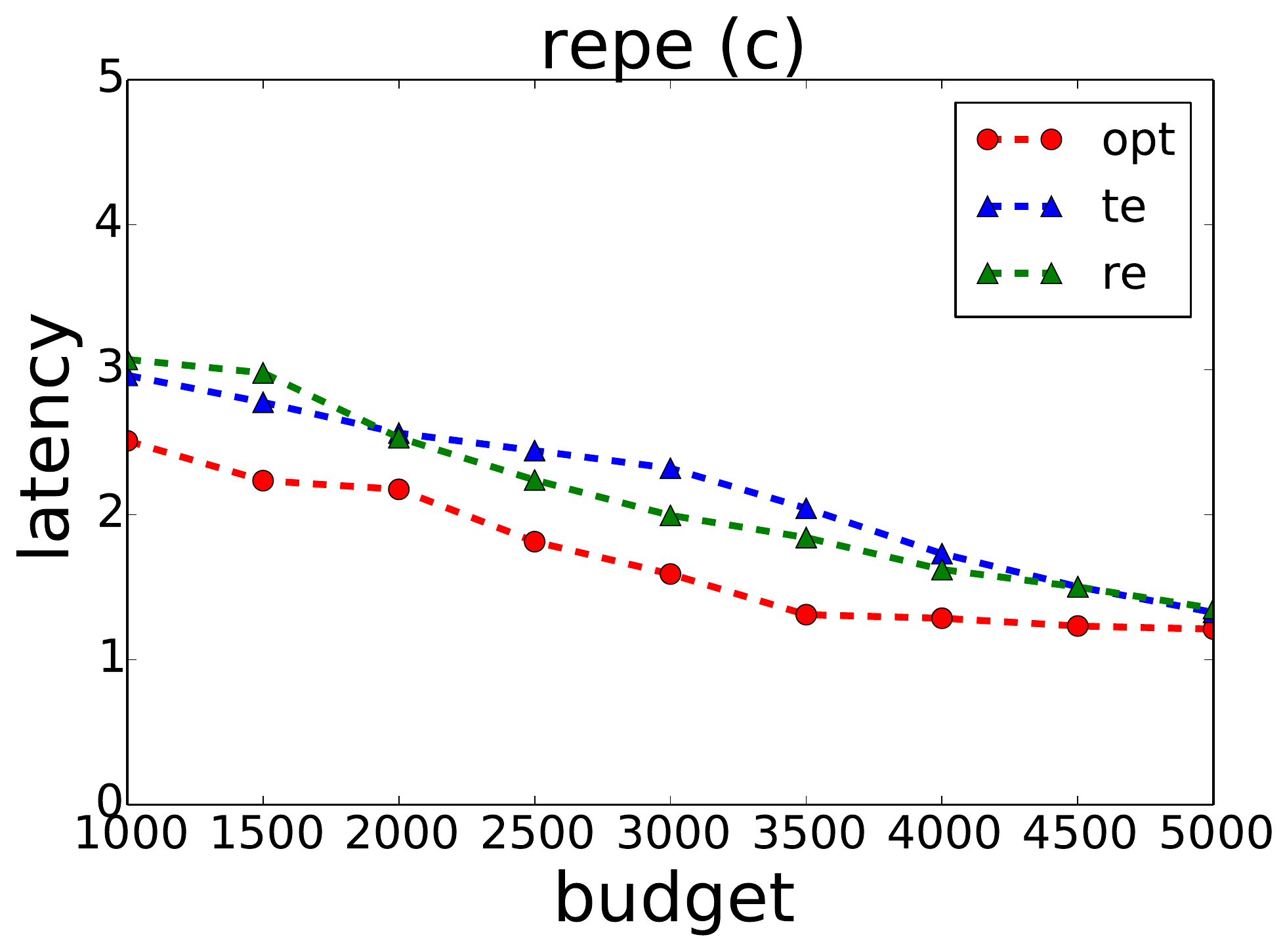}
}
\subfigure[Repetition($\lambda=1+p^{2}$)] { \label{fig:11}
\includegraphics[height = 1.0in,width=0.63\columnwidth]{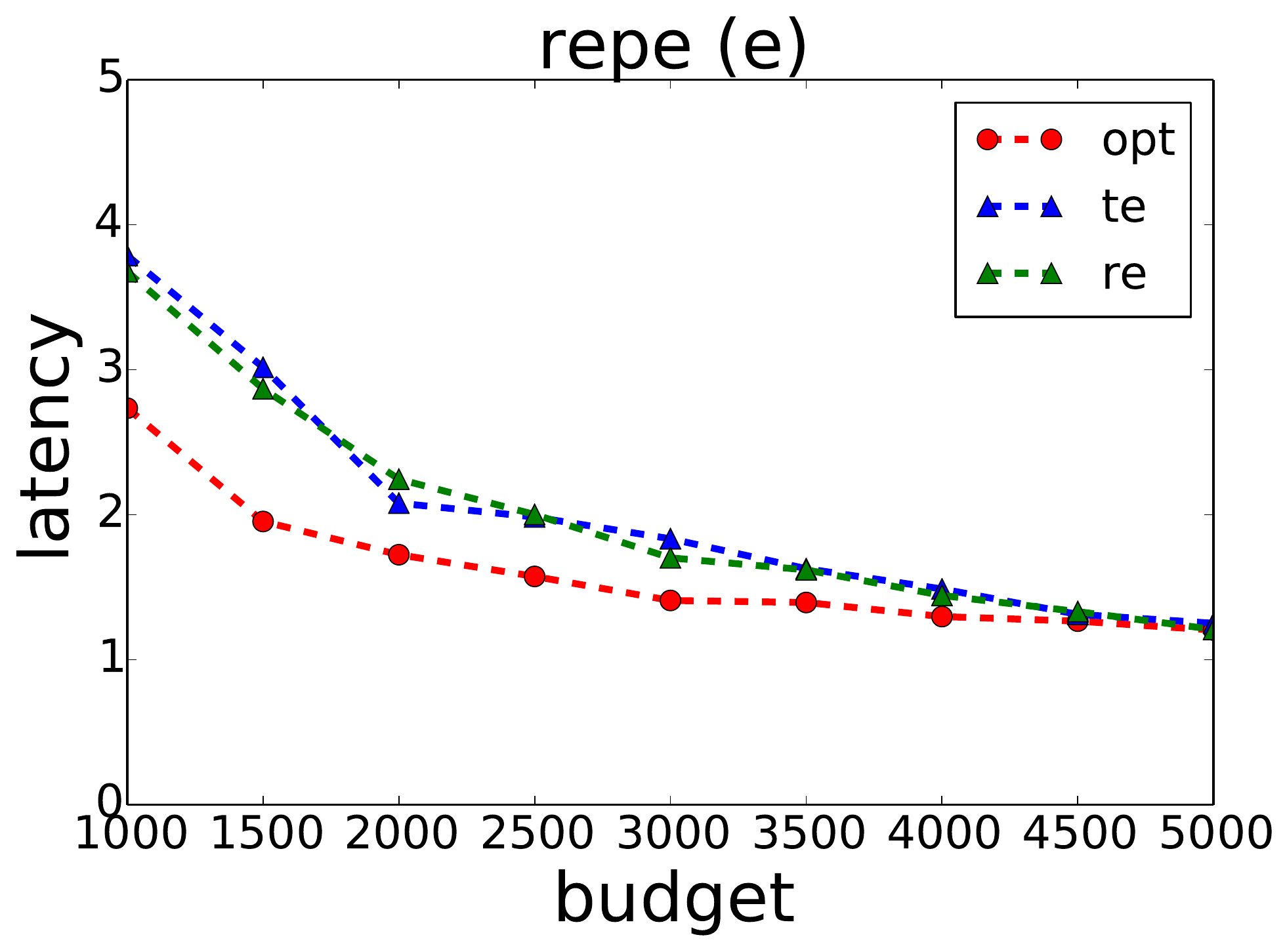}
}
\subfigure[Repetition($\lambda=\log{1+p}$)] { \label{fig:12}
\includegraphics[height = 1.0in,width=0.63\columnwidth]{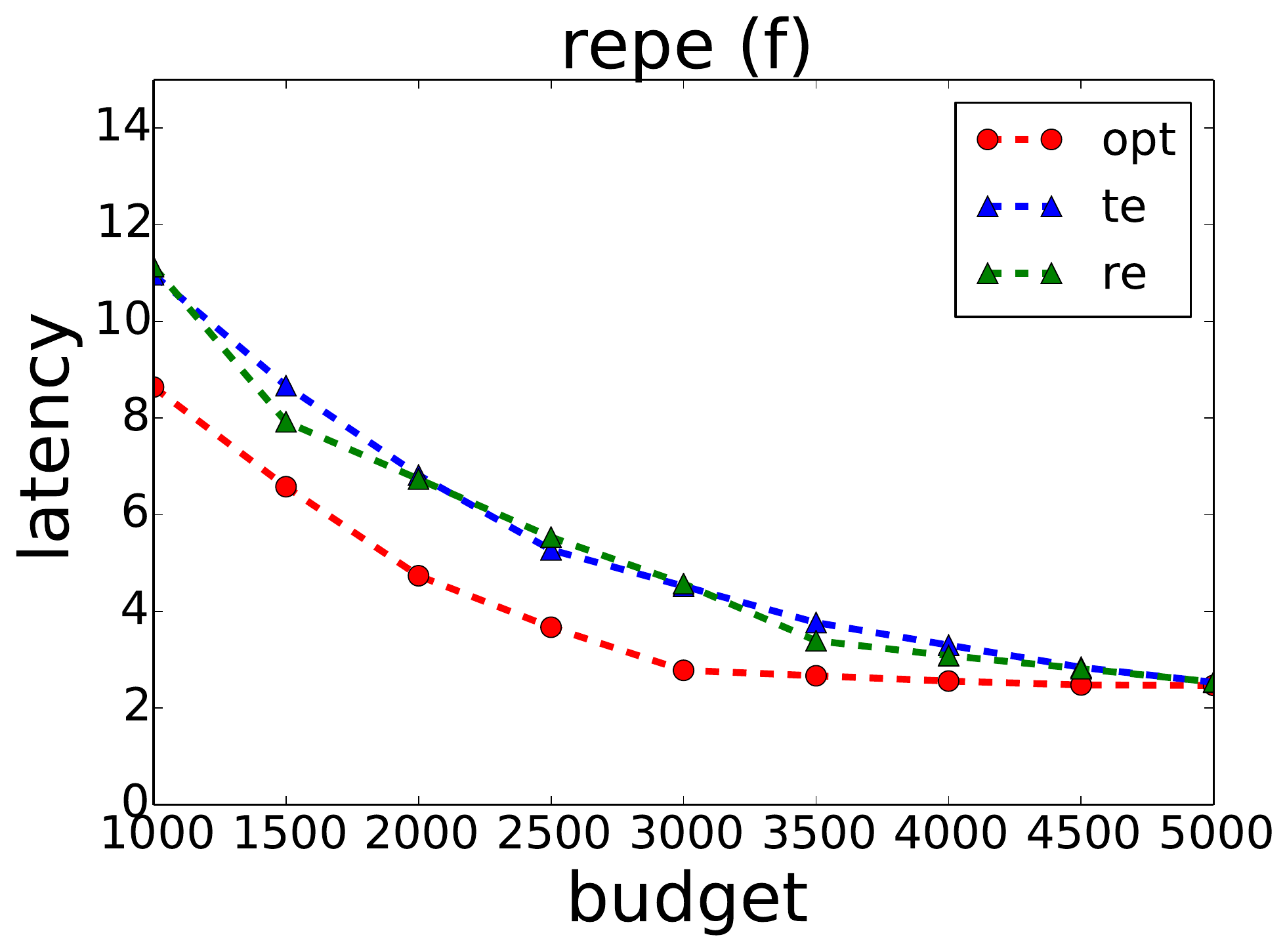}
}
\subfigure[Heterogeneous($\lambda=1+p$)] { \label{fig:13}
\includegraphics[height = 1.0in,width=0.63\columnwidth]{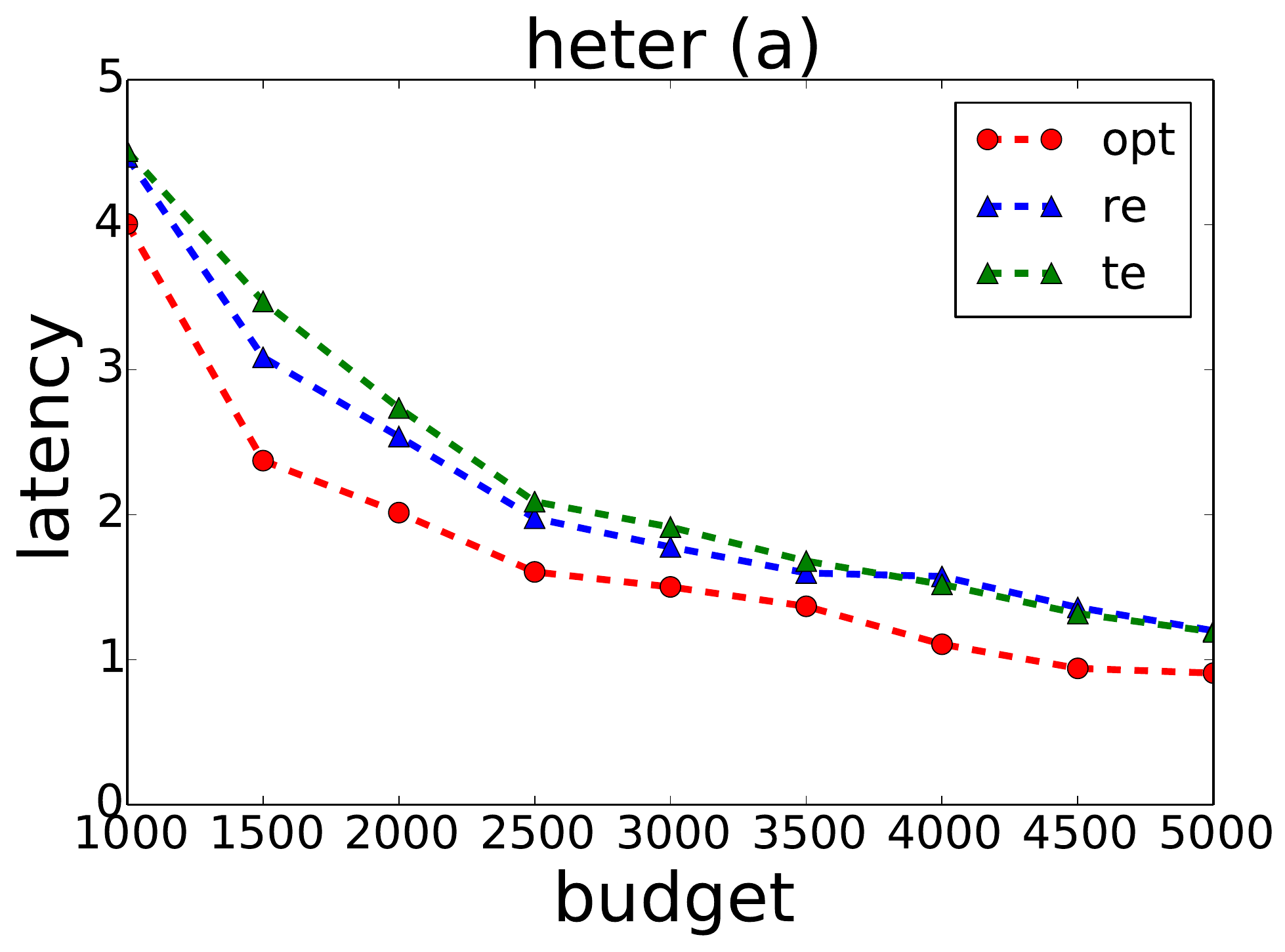}
}
\subfigure[Heterogeneous($\lambda=10p+1$)] { \label{fig:14}
\includegraphics[height = 1.0in,width=0.63\columnwidth]{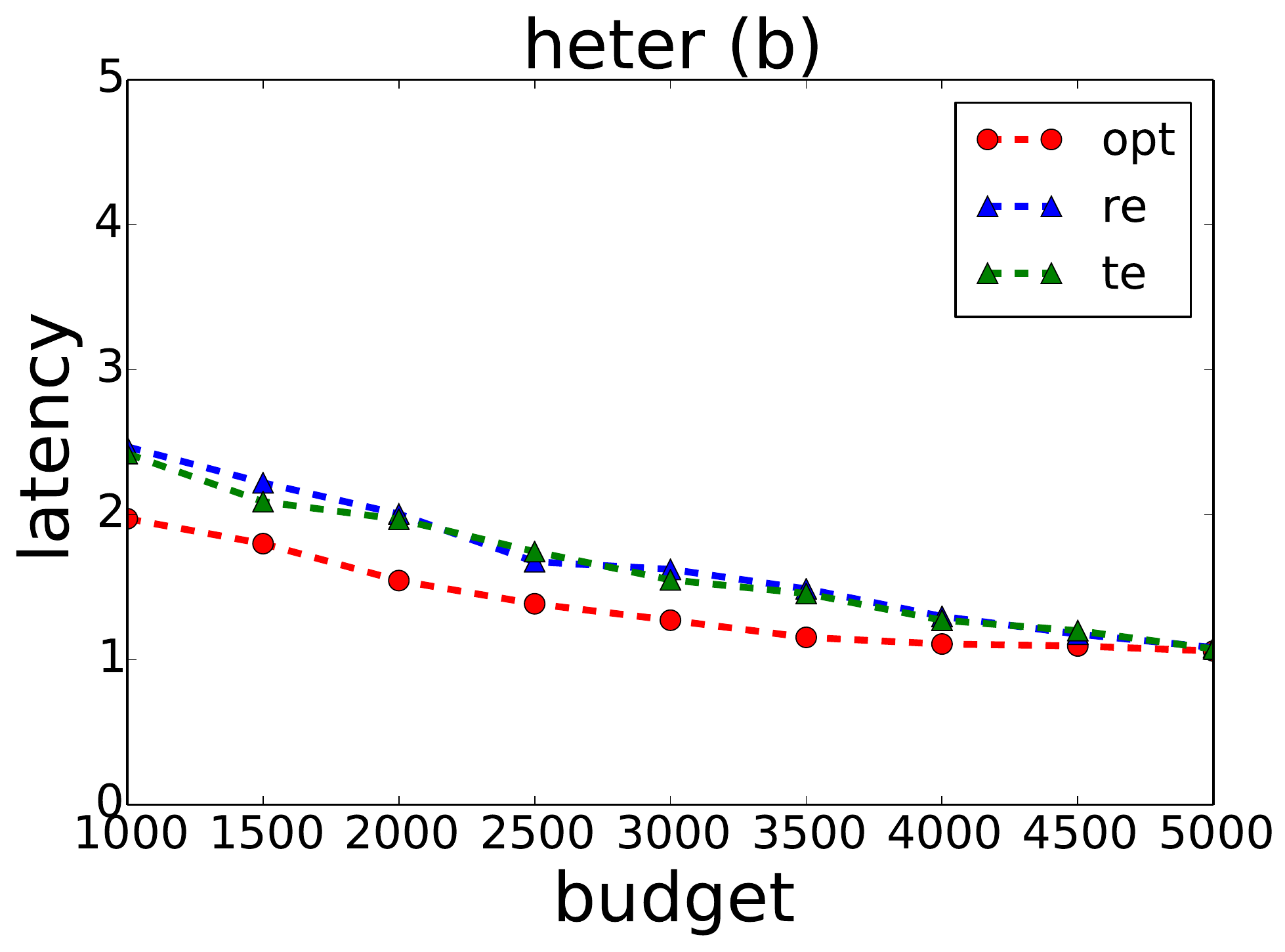}
}
\subfigure[Heterogeneous($\lambda=0.1p+10$)] { \label{fig:15}
\includegraphics[height = 1.0in,width=0.63\columnwidth]{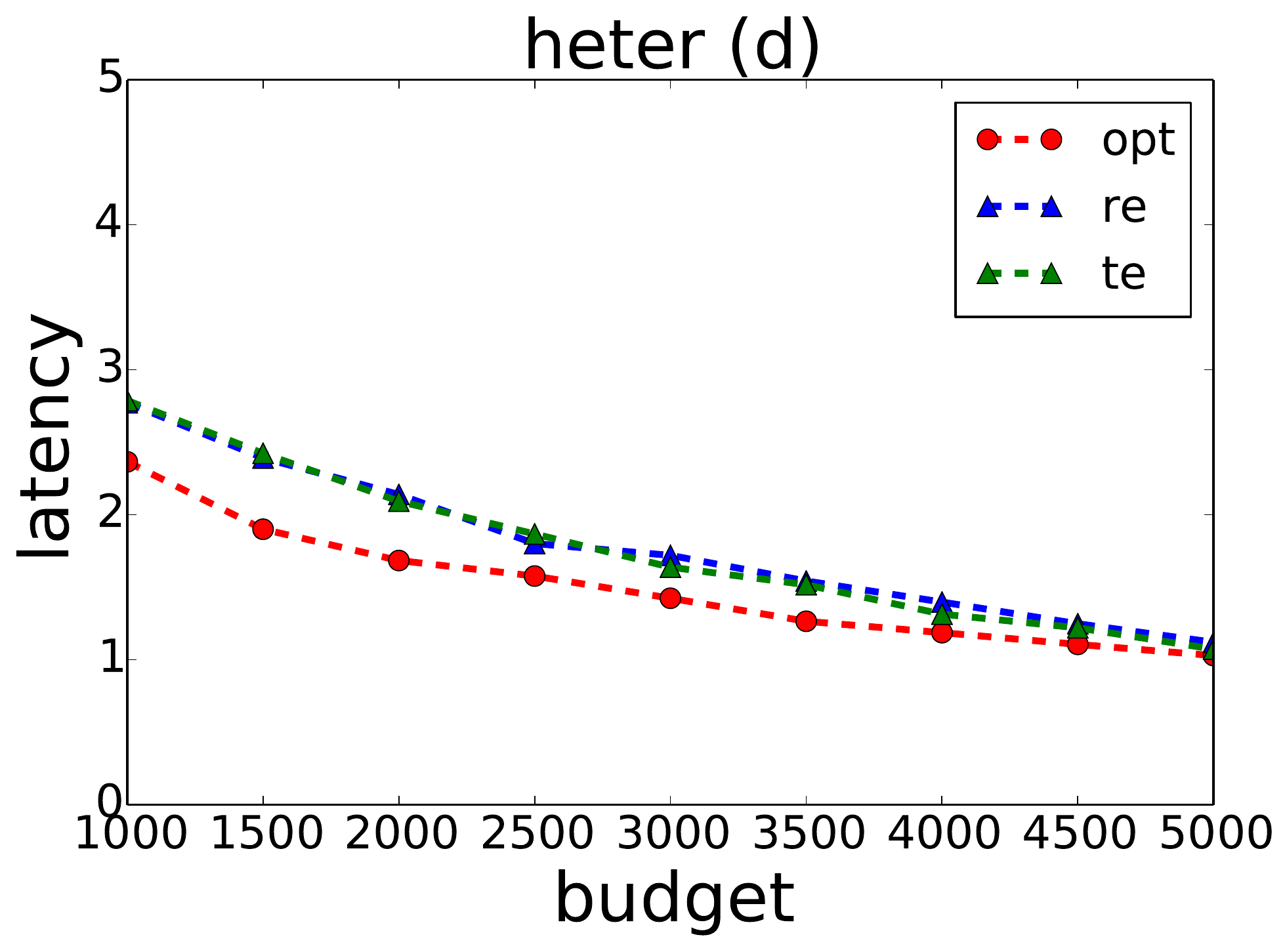}
}
\subfigure[Heterogeneous($\lambda=3p+3$)] { \label{fig:16}
\includegraphics[height = 1.0in,width=0.63\columnwidth]{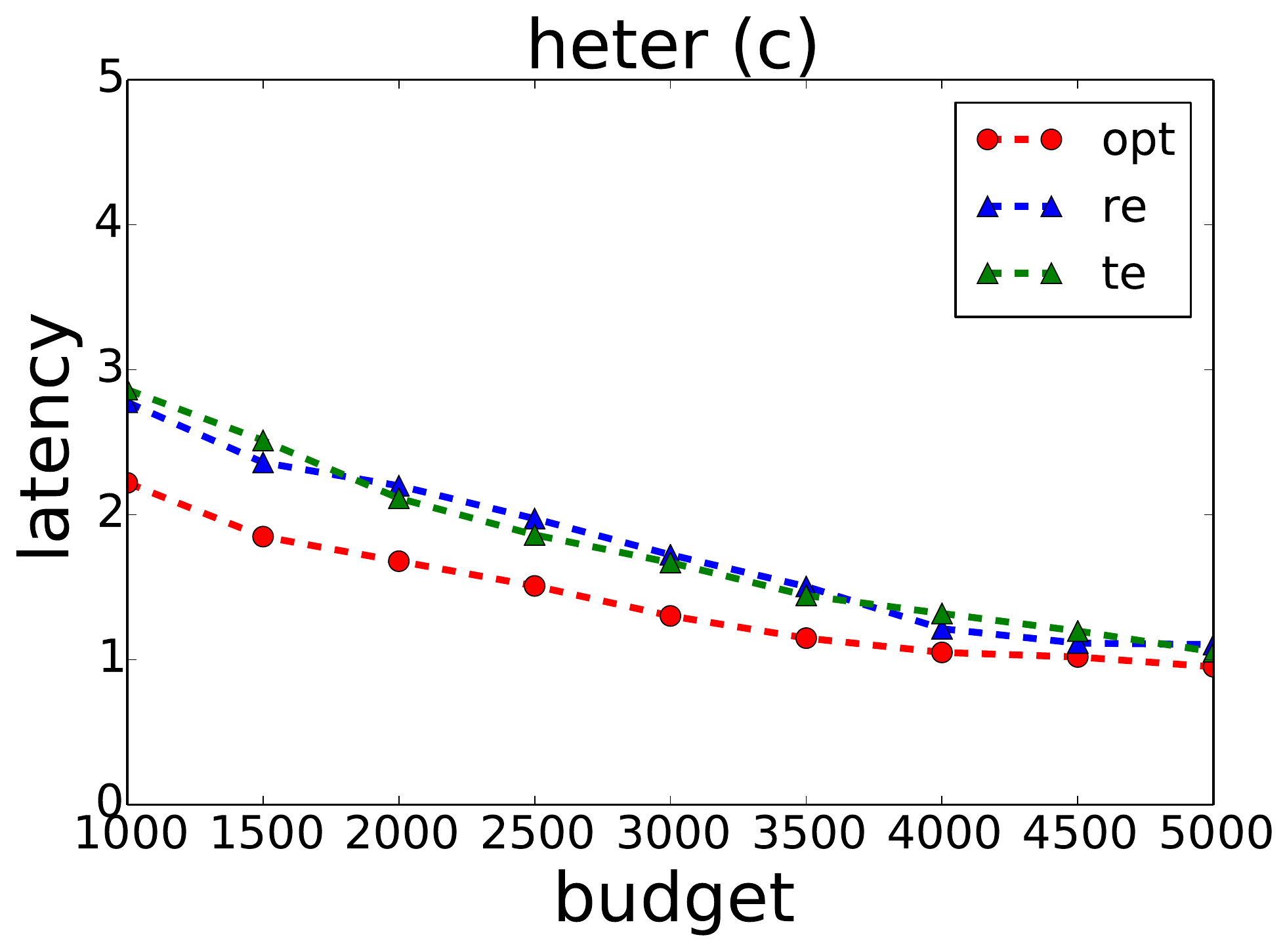}
}
\subfigure[Heterogeneous($\lambda=1+p^{2}$)] { \label{fig:17}
\includegraphics[height = 1.0in,width=0.63\columnwidth]{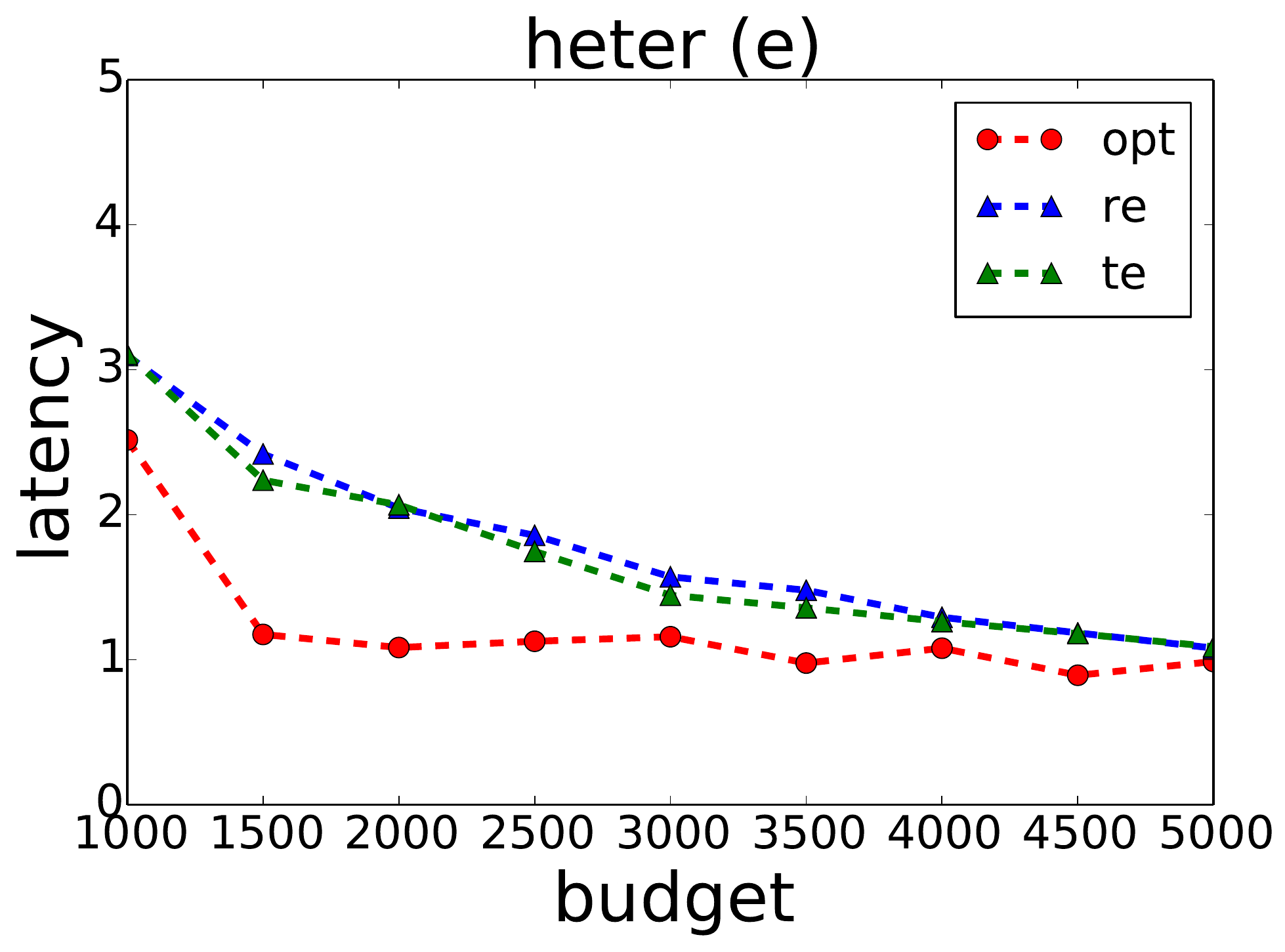}
}
\subfigure[Heterogeneous($\lambda=\log{1+p}$)] { \label{fig:18}
\includegraphics[height = 1.0in,width=0.63\columnwidth]{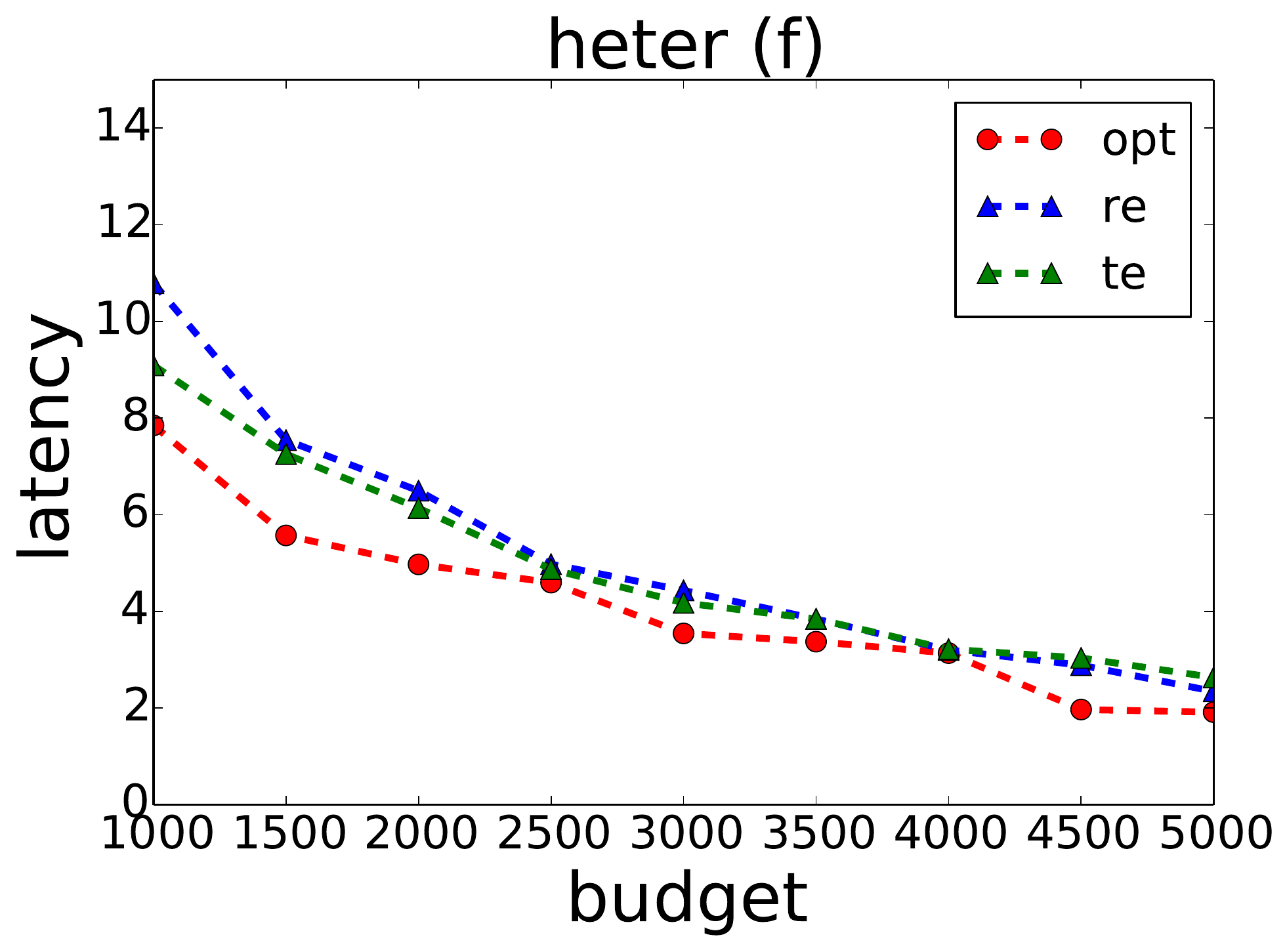}
}
\caption{Experiments on Synthetic Data}
\end{figure*}

\vspace{-1em}
%%%%%%%%%%%%%%%%%Experiments%%%%%%%%%%%%%%%%%%%
\section{Experiments}\label{section:4experiments}
We extensively evaluated our model and optimization techniques, and report on the results here.
While the gold standard is performance on a real platform, we can exercise greater control and thereby get a better empirical understanding of our system through simulation.
Therefore, we did both.  We report first on simulation results with synthetic data, and then on jobs executed on Amazon Mechanical Turk.

\subsection{HPU Traits Testing with Synthetic Data}
\subsubsection{Experiments Settings}

We conduct six sets of experiment for each Scenario. The first four sets are linear model based, which aim to verify the effectiveness of the tuning strategy under the linear Hypothesis, and the last two sets are nonlinear model based, which aim to test the robustness of the tuning strategy. For the linear model based experiment, the model parameters are set as $\lambda^o=p+1$, $\lambda^o=10p+1$, $\lambda^o_=0.1p+10$, $\lambda^o=3p+3$. For the nonlinear part, the parameters are set as $\lambda^o=1+p^2$ and $\lambda^o=\log{(1+p)}$. The total number of task is set to be 100 uniformly for each set of the experiment and the budget varies from 1000 to 5000. 

\textbf{Homogeneity} All tasks call fro 5 repetitions. As the difficulty of the tasks are identical, the clock rate $\lambda_{p}$ for the processing latency is uniformly set to be $2.0$. Since the optimistic solution is produced by the \emph{even allocation} (\emph{algorithm 1}), biased allocation strategies are adopted as the baseline comparison. Instead of allocating the budget evenly, the biased method gives more payment to one half of the tasks, while less payment to the other. Specifically, half of the tasks are randomly selected as ``the prior group'' which take up $\alpha$ ($\frac{1}{2} < \alpha < 1$) of the total budget ($\alpha=\frac{1}{2}$ leads to the even allocation), and the remaining tasks get the $1-\alpha$ of the total budget. The value of the alpha is set to be $0.67$, and $0.75$ in our experiment.

\textbf{Repetition} The tasks are divided equally into two groups: one group is of 3 repetitions for each task, while the other group is of 5 repetitions for each of the tasks. Still, $\lambda_{p}$ is uniformly set to be $2.0$ due the identical setting of the difficulty. Two baseline methods are chosen as for the comparison. The first method is called \textit{task-even} allocation, which gives identical price to each task, then every task allocate the total budget evenly to each of its repetitions. Therefore, repetition price for group 2 is 60\% of that of group 1. The second one is called \textit{rep-even} allocation, which gives identical price to each repetition of all the tasks. So the total price for the tasks in the group 1 is 60\% of that of group 2.   

\textbf{Heterogeneous} The tasks are dived into two groups: task in the first group call fro three repetitions, while tasks in the second group call fro five repetitions. 
Then $\lambda_p$ is set to be $2.0$ and $3.0$ these two groups, respectively. Same with scenario \rom{2}, the \textit{rep-even} and \textit{task-even} are chosen the baseline methods for the comparison.  

We also conduct experiments with different settings of the budget, task amount, repetitions, and difficulty. However, there's no significant variance between different settings. Therefore, we simply demonstrate the results of the above setting for further analysis.

\subsubsection{Results Summary}
From the experiment results, the optimal solution outperforms the comparisons in terms of latency in every cases. For results of scenario 1 (homo), the ``\emph{bias\_1}'' produces slightly better performance than ``\emph{bias\_2}''. This is because \emph{bias\_2} is more biased (the value of $\alpha$ is larger) than \emph{bias\_1}. Such phenomenon further verified our conclusion that even allocation leads to the optimal budget plan for the Scenario 1. Besides, we can find that although the optimal solution of the Scenario 1 is designed based on the linear hypothesis, it still works for the nonlinear cases (\emph{homo(e)} and \emph{homo(f)}), which can be partially explained by the varying range of the payment: the task price varies form 1 to 9. For such relatively low prices, the non-linear relationship can be linearly approximated quite well.
We can further find that the optimal results are relatively close to the comparison in case (b) and (c), for all the scenarios. For case (b), such phenomenon can be caused the large value of the linear coefficient ($lambda=10p+1$). When the linear coefficient is large, the on-hold clock rate is sensitive to the change of price. When price grows, the clock rate increases much more faster, making the oh-hold latency decrease to a low level with a relatively lower price. In this situation, the overall latency will be mostly determined by the processing phase. Similar phenomenons can be observed for case (e) in each of the scenarios, where overall latency reduces sharply for the initial prices, and soon get to a stable level. While case (c) is another extreme, where $\lambda$ is fairly insensitive to the price changes. In this situation, and the latency is largely determined by the initial setting of on-hold and processing phase, and price does little to change it.

Finally, we can summarize the findings of the synthetic experiment as follows: 1) the optimal tuning strategy is robust to non-linearity. The unit price for each task is usually small, therefore the linearity hypothesis holds for normal cases. 2) The optimal tuning strategy is sensitive to the \emph{price-}$lambda$ relationship: when $lambda$ is sensitive to the change of price, the on-hold latency drops sharply with the growing price. Then the overall latency is determined by the processing time and it's unnecessary to keep on increasing the price.  

%
%\begin{figure}
%\centering \includegraphics[height = 0.9in,width=0.75\columnwidth]{fig19.eps}
%\caption{Enhancement with Scalability} \label{fig:19}
%\end{figure}

\subsection{Tuning Tasks on Amazon MTurk}
\vspace{-0.5em}
\begin{figure}
\centering \includegraphics[height = 0.9in,width=0.75\columnwidth]{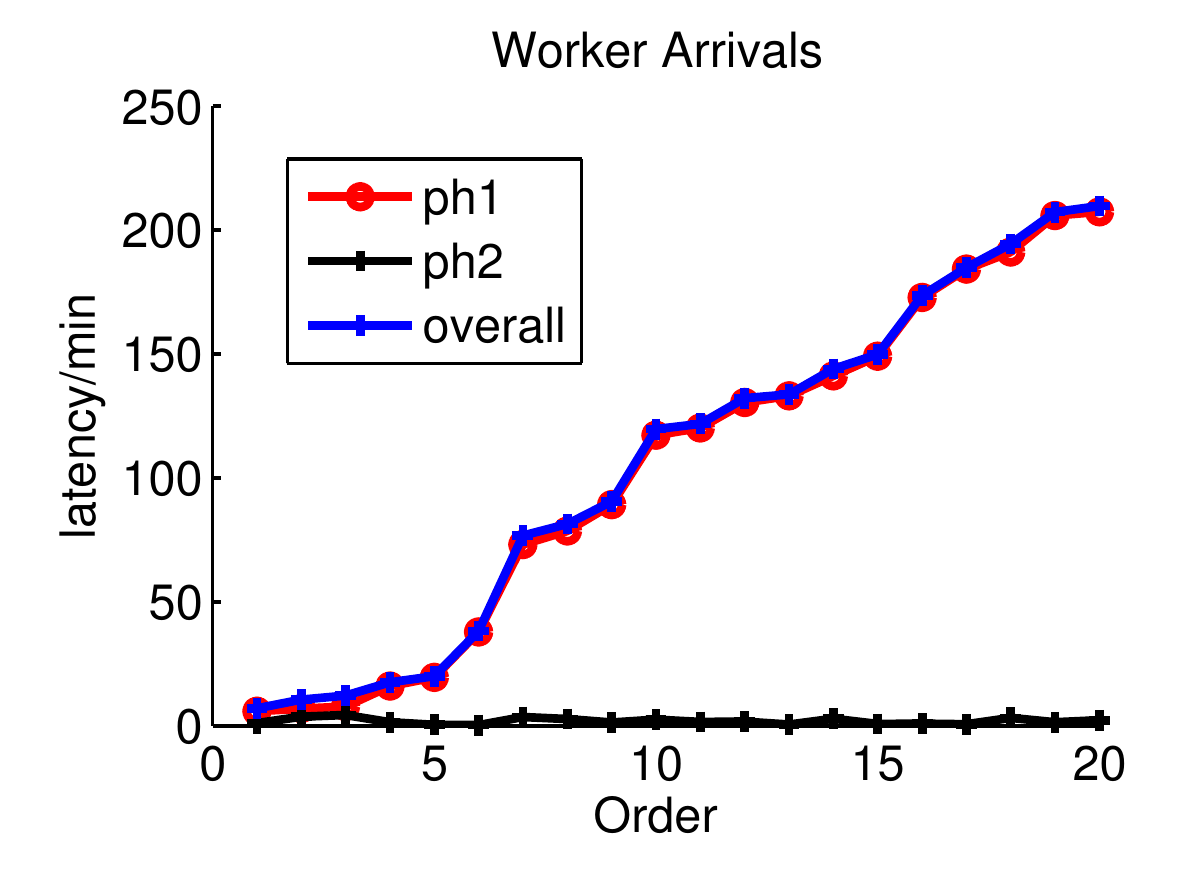}
\caption{Worker Arrival Moments} \label{fig:20}
\end{figure}

\begin{figure}
\centering \includegraphics[height = 0.9in,width=0.75\columnwidth]{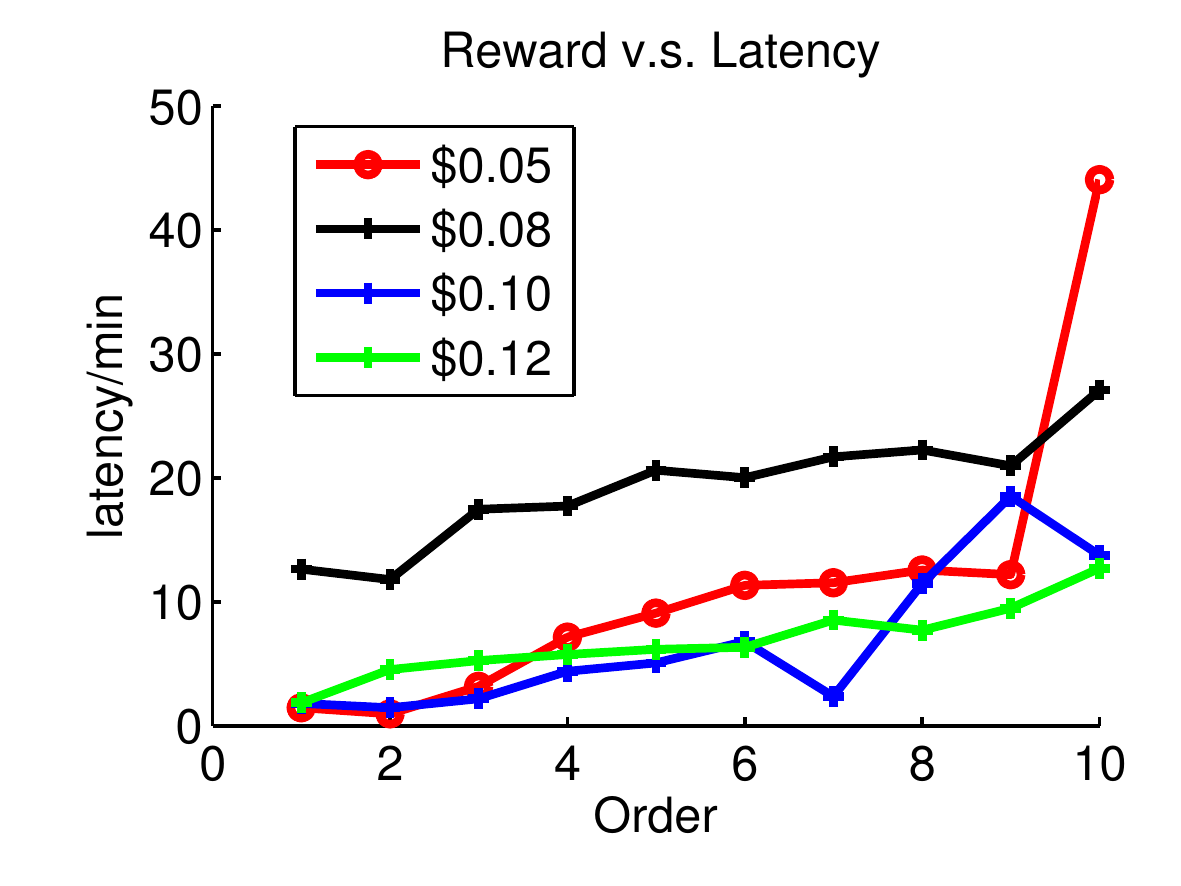}
\caption{Money v.s. Latency} \label{fig:21}
\end{figure}

\begin{figure*}[htbp] \centering
\subfigure[Difficulty v.s. Phase 1] { \label{fig:22}
\includegraphics[height = 0.9in,width=0.63\columnwidth]{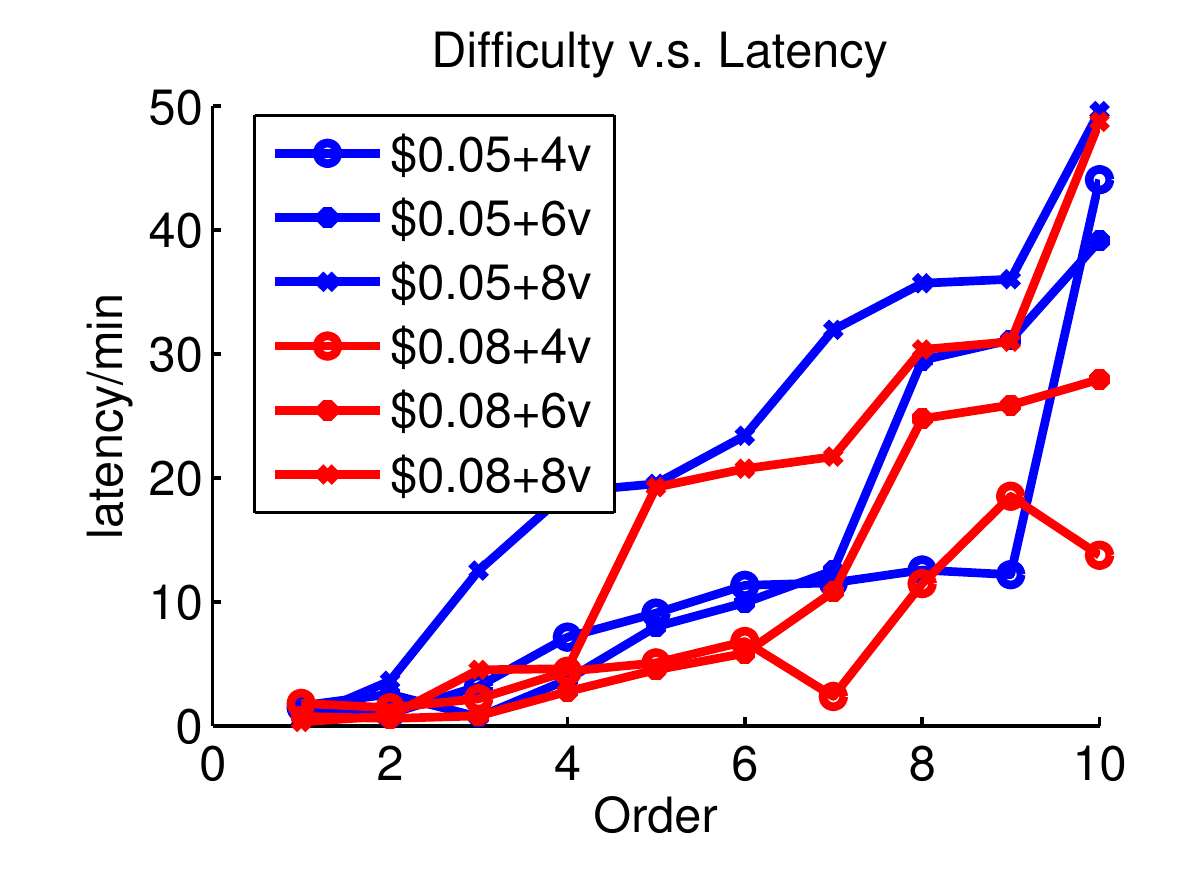}
}
\subfigure[Difficulty v.s. Phase 2] { \label{fig:23}
\includegraphics[height = 0.9in,width=0.63\columnwidth]{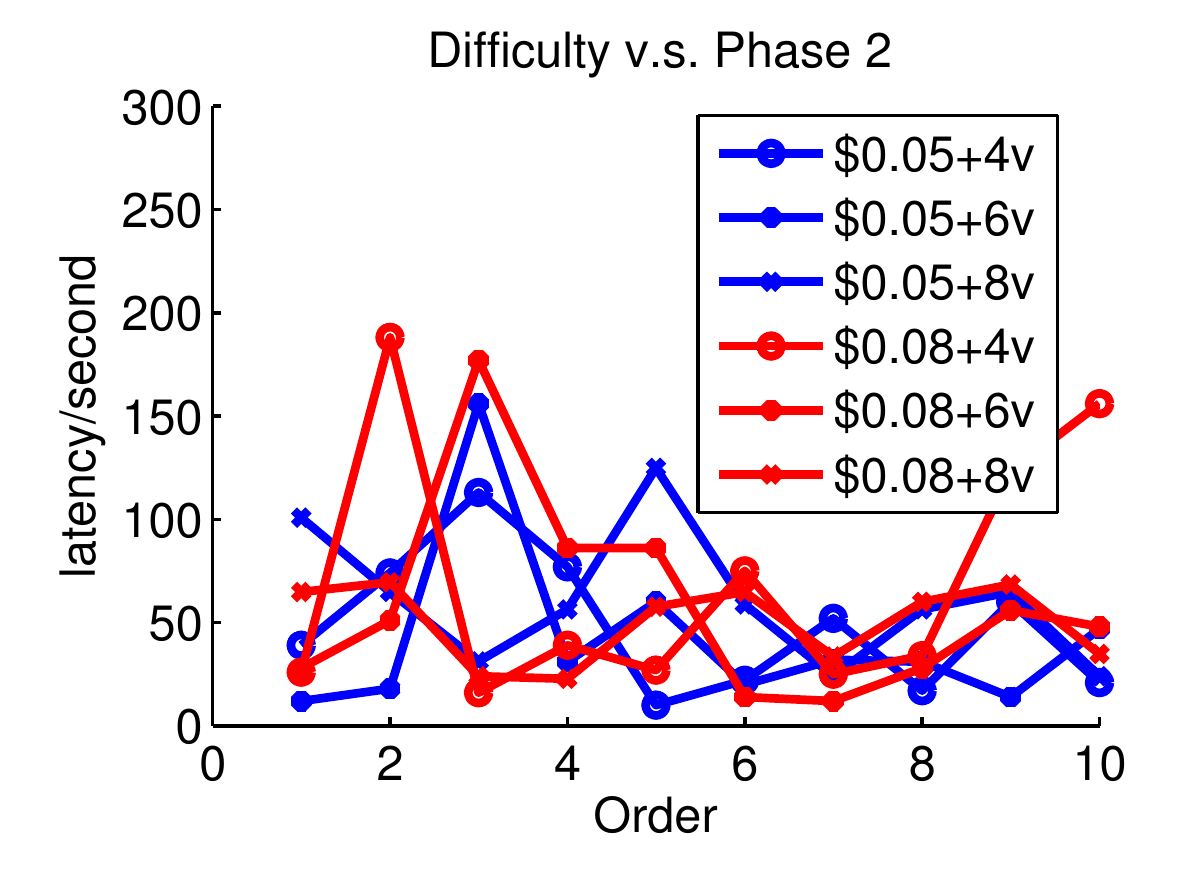}
}
\subfigure[OPT v.s. Heuristic] { \label{fig:24}
\includegraphics[height = 0.9in,width=0.63\columnwidth]{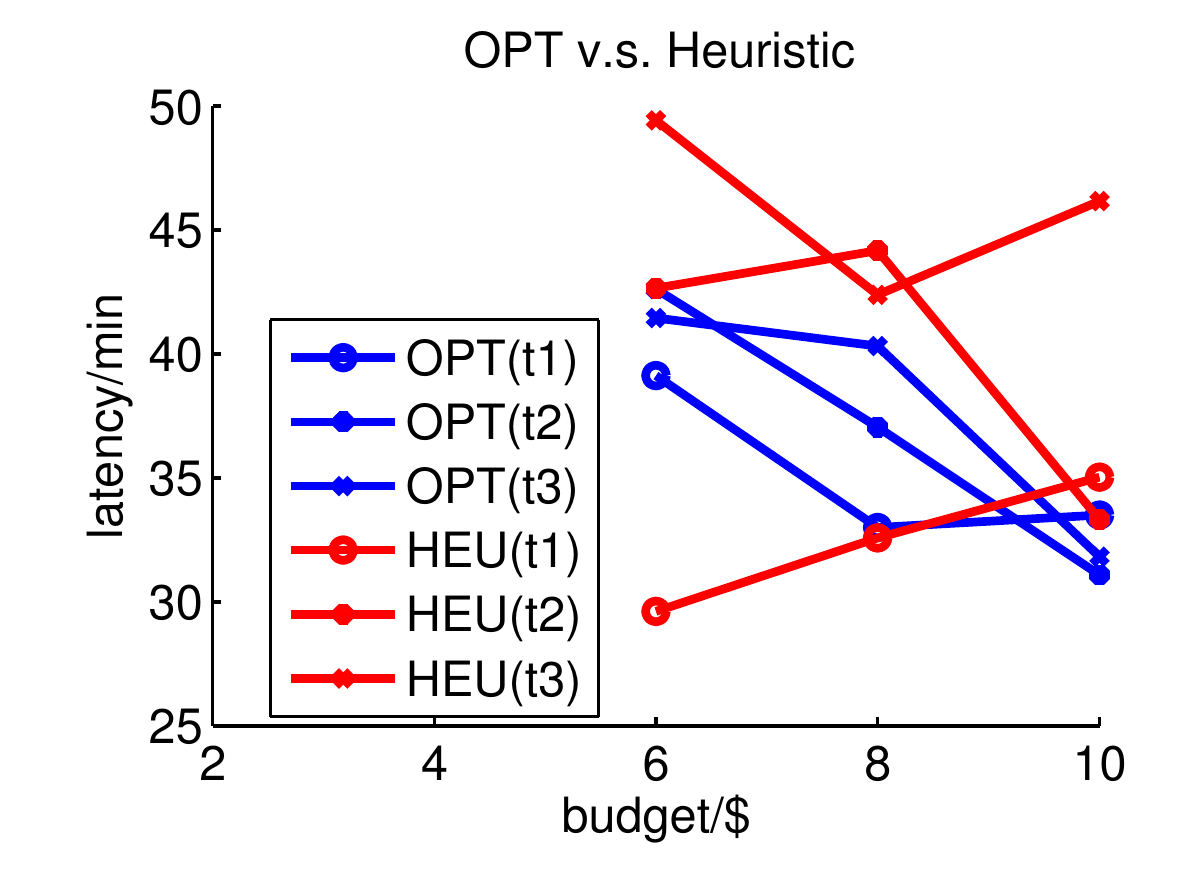}
}
\caption{Experiments Results II}\label{fig:II}
\end{figure*}
\subsubsection{Experiments Settings}
We create a set of image filtering tasks as the atomic tasks: we first present the workers an image with the exact number of the dots on it, then a set of images are presented to the ``workers'' and they are required to estimate the number of dots on each image. Based on the estimation, ``workers'' are expected to filter out the ones who have dots less than a given threshold. Under such settings, the cognitive abilities of ``recognizing'' and ``counting'' are utilized, and the task is finished by presenting a set of binary voting(clicking on the checkbox). In addition, the ``workers'' receive their rewards when the provided answers are correct. We control the difficulty or type of tasks by varying the images given in a single tasks.

Our work focuses on tuning upon budget allocation and real time latency, thus we purposely design the experiment simple enough and avoid setting any worker qualifications and inter-rater agreement. In fact, in real scenarios, the atomic tasks on lowest level are just the same as the experimental tasks: comparing items, screening out candidates and simple ranking.

%\textbf{APIs for HPU Tuning} Most industrial crowdsourcing platforms provide programming APIs to help manage the publishing, answers collecting and worker resource control. Specifically in the tuning operations, the following APIs are employed(we adopt Java SDK, updated at Feb 2013):
%\begin{itemize}
%\item[-]  \textit{createHIT(title, description, reward, question, maxAssignments})$\rightarrow$ We specify the optimized budget allocation by tailoring the ``reward'' argument. The atomic tasks are described in ``question'' as XML input;
%\item[-] \textit{getCreationTime()} $\rightarrow$ This API is bonded with an HIT object, where we can record the starting epoch of the first phase;
%\item[-] \textit{getAcceptTime()} $\rightarrow$ This API is bonded with an \textit{Assignment} object, which serves as the container of collected answers and is able to record the ending epoch of the first phase, i.e. the starting epoch of the second phase;
%\item[-] \textit{getSubmitTime()} $\rightarrow$ This API is also bonded with an \textit{Assignment} object, and it provides the ending epoch of the second phase.
%\end{itemize}

\subsubsection{Results Summary}
%Please refer to Figures~\ref{fig:20} to ~\ref{fig:24} for the experiments conducted on Amazon MTurk.

Firstly, in Fig~\ref{fig:20} we present the general behavior of the worker appearance and the latency of processing time. We issue image filtering tasks with 1 unit reward(\$0.05), and collect the first 20 arrivals. As shown in Fig~\ref{fig:20}, the arrival epochs of the workers exhibit linearity, indicating the suitability of the Poisson Process Model, while the latency of the second phase fluctuates in a small range.

Then we examine the effect of varying the rewards: we vary the reward on a single task from \$0.05 to \$0.12, while for each task we require 10 repetitions. The results can be found in Fig~\ref{fig:21}, where it is obviously that the increase on rewards incurs shorter latencies. According to the methodologies introduced in Section~\ref{section:running_para}, we obtain the corresponding parameters($s^{-1}$), $\lambda_{1}=0.0038,\lambda_{2}=0.0062,\lambda_{3}=0.0121,\lambda_{4}=0.0131$, which supports the Linearity Hypothesis proposed in Section~\ref{section:hypo}.

In the sequel, we present the results of examining the effect of
varying the type of the tasks: we vary the internal binary voting
number from 4 to 8. Such change of difficulty results in the
decrease of the coming rate (see Fig~\ref{fig:22}), and
the increase of the average processing time, which is shown
in Fig~\ref{fig:23}. We then evaluate our proposed algorithms on
Amazon MTurk, especially under Scenario II and III. Namely, 3 types
of tasks are published with different repetition requirement: 10 for
$t1$, 15 for $t2$ and 20 for $t3$. The total budgets are also varied
from \$6 to \$10. We compare our algorithms(OPT) with the heuristic
where each type receives same payment under both two scenarios.
Results can be found in Fig~\ref{fig:24}, where the lower latency of
OPT shows the effectiveness of our algorithms. Note that at each
budget, the OPT successfully avoids yielding the longest
latency among the three tasks.

%%%%%%%%%%%%%%%%%%%Conclusion%%%%%%%%%%%%%%
\section{Conclusion}\label{section:6conclusion}
In this paper, we address the problem of tuning the modularized
human computation, so that the latency in real clock time could be
minimized. The difficulty of such problem arises in the stochastic
behavior of the latency of the HPU. To address this challenge, we
theoretically and practically propose that appearance of the crowd
``workers'' follows a Poisson Process, whose parameter differs at
different budget levels and types of atomic tasks. Then we formally
propose the \textit{H-Tuning Problem} to optimize the expected
latency of the longest task. Moreover, under three most general scenarios on crowd-powered applications, advanced strategies
are designed to cope with the \textit{H-Tuning Problem}. Finally, a series of experiments conducted on both simulated data and real commercial platform observe the effectiveness
of the proposed model and strategies. To conclude, the crowdsourced
human computation is now equipped with primitive tuning ability in
terms of running time.

%Considering the HPU a fundamental computation module, we expect more
%sophisticated tuning strategies:
%\begin{inparaenum}[\itshape i\upshape)]
%  \item instead of the one-shot tuning, we will further enhance the throughput of the HPU to support more robust applications;
%  \item while allowed to reform atomic tasks by changing the form of voting, we will further study the HPU's ``batching'' behavior;
%  \item the quality of answers varies according to the difficulty and nature of the tasks, whereas the reward only affect coming rate of the ``workers'', thus the tradeoff between the tuning power and answer quality entails thorough study.
%\end{inparaenum}

%%%%%%%%%%%%%%%%%%%Reference%%%%%%%%%%%%%%%%%%%%%
\begin{small}
\bibliographystyle{IEEEtran}
\bibliography{hpu}

\begin{thebibliography}{10}
\providecommand{\url}[1]{#1}
\csname url@samestyle\endcsname
\providecommand{\newblock}{\relax}
\providecommand{\bibinfo}[2]{#2}
\providecommand{\BIBentrySTDinterwordspacing}{\spaceskip=0pt\relax}
\providecommand{\BIBentryALTinterwordstretchfactor}{4}
\providecommand{\BIBentryALTinterwordspacing}{\spaceskip=\fontdimen2\font plus
\BIBentryALTinterwordstretchfactor\fontdimen3\font minus
  \fontdimen4\font\relax}
\providecommand{\BIBforeignlanguage}[2]{{%
\expandafter\ifx\csname l@#1\endcsname\relax
\typeout{** WARNING: IEEEtran.bst: No hyphenation pattern has been}%
\typeout{** loaded for the language `#1'. Using the pattern for}%
\typeout{** the default language instead.}%
\else
\language=\csname l@#1\endcsname
\fi
#2}}
\providecommand{\BIBdecl}{\relax}
\BIBdecl

\bibitem{davis:hpu10}
J.~Davis, J.~Arderiu, H.~Lin, Z.~Nevins, S.~Schuon, O.~Gallo, and M.~Yang,
  ``The hpu,'' in \emph{CVPRW 2010}.

\bibitem{law:human11}
\emph{Human Computation}.

\bibitem{grier:ieeeann98}
D.~A. Grier, ``The math tables project of the work projects administration: The
  reluctant start of the computing era,'' \emph{IEEE Ann. Hist. Comput.}

\bibitem{aditya:scoop11}
A.~Parameswaran and N.~Polyzotis, ``Answering queries using humans, algorithms
  and databases,'' in \emph{CIDR 2011}.

\bibitem{michael:sigmod11:crowddb}
M.~J. Franklin, D.~Kossmann, T.~Kraska, S.~Ramesh, and R.~Xin, ``Crowddb:
  answering queries with crowdsourcing,'' in \emph{SIGMOD 2011}.

\bibitem{marcus:vldb11}
A.~Marcus, E.~Wu, D.~Karger, S.~Madden, and R.~Miller, ``Human-powered sorts
  and joins,'' \emph{VLDB 2011}.

\bibitem{aditya:sigmod12}
A.~G. Parameswaran, H.~Garcia-Molina, H.~Park, N.~Polyzotis, A.~Ramesh, and
  J.~Widom, ``Crowdscreen: algorithms for filtering data with humans,'' in
  \emph{SIGMOD 2012}.

\bibitem{venetis:www12}
P.~Venetis, H.~Garcia-Molina, K.~Huang, and N.~Polyzotis, ``Max algorithms in
  crowdsourcing environments.''

\bibitem{guo:sigmod12}
S.~Guo, A.~G. Parameswaran, and H.~Garcia-Molina, ``So who won?: dynamic max
  discovery with the crowd,'' in \emph{SIGMOD Conference 2012}.

\bibitem{milo:icdt13}
S.~B. Davidson, S.~Khanna, T.~Milo, and S.~Roy, ``Using the crowd for top-k and
  group-by queries,'' in \emph{ICDT 2013}.

\bibitem{reynold:icde13:tagging}
X.~S. Yang, R.~Cheng, L.~Mo, B.~Kao, and D.~W. Cheung, ``On incentive-based
  tagging,'' in \emph{ICDE 2013}.

\bibitem{jason:vldb13}
J.~C. Zhang, L.~Chen, H.~V. Jagadish, and C.~C. CAO, ``Reducing uncertainty of
  schema matching via crowdsourcing,'' \emph{VLDB 2013}.

\bibitem{jiannan:sigmod13}
J.~Wang, G.~Li, T.~Kraska, M.~J. Franklin, and J.~Feng, ``Leveraging transitive
  relations for crowdsourced joins,'' in \emph{SIGMOD 2013}.

\bibitem{milo:icde13}
I.~Lotosh, T.~Milo, and S.~Novgorodov, ``Crowdplanr: Planning made easy with
  crowd,'' in \emph{ICDE 2013}.

\bibitem{milo:sigmod13}
Y.~Amsterdamer, Y.~Grossman, T.~Milo, and P.~Senellart, ``Crowd mining,'' in
  \emph{SIGMOD 2013}.

\bibitem{jwang:csdm11}
J.~Wang, S.~Faridani, and P.~Ipeirotis, ``Estimating the completion time of
  crowdsourced tasks using survival analysis models,'' \emph{CSDM 2011}.

\bibitem{faridani:hcomp11}
S.~Faridani, B.~Hartmann, and P.~G. Ipeirotis, ``What's the right price?
  pricing tasks for finishing on time.'' in \emph{Hcomp 2011}.

\bibitem{mason:hcomp09}
W.~Mason and D.~J. Watts, ``Financial incentives and the "performance of
  crowds",'' ser. HCOMP 2009.

\bibitem{yan:mobisys10}
T.~Yan, V.~Kumar, and D.~Ganesan, ``Crowdsearch: exploiting crowds for accurate
  real-time image search on mobile phones,'' in \emph{MobiSys 2010}.

\bibitem{ooi:vldb12:cdas}
X.~Liu, M.~Lu, B.~C. Ooi, Y.~Shen, S.~Wu, and M.~Zhang, ``Cdas: a crowdsourcing
  data analytics system,'' \emph{VLDB 2012}.

\bibitem{caleb:vldb12}
C.~C. CAO, J.~She, Y.~Tong, and L.~Chen, ``Whom to ask? jury selection for
  decision making tasks on micro-blog services,'' \emph{VLDB 2012}.

\bibitem{ipei:hcomp10:quality}
P.~G. Ipeirotis, F.~Provost, and J.~Wang, ``Quality management on amazon
  mechanical turk,'' in \emph{Proceedings of the ACM SIGKDD Workshop on Human
  Computation}.

\bibitem{milo:icde12}
R.~Boim, O.~Greenshpan, T.~Milo, S.~Novgorodov, N.~Polyzotis, and W.~C. Tan,
  ``Asking the right questions in crowd data sourcing,'' in \emph{ICDE 2012}.

\bibitem{little:uist10:vizwiz}
J.~P. Bigham, C.~Jayant, H.~Ji, G.~Little, A.~Miller, R.~C. Miller, R.~Miller,
  A.~Tatarowicz, B.~White, S.~White, and T.~Yeh, ``Vizwiz: nearly real-time
  answers to visual questions,'' in \emph{UIST 2010}.

\bibitem{galen:sci11:mobilization}
G.~Pickard, W.~Pan, I.~Rahwan, M.~Cebrian, R.~Crane, A.~Madan, and A.~Pentland,
  ``Time-critical social mobilization,'' \emph{Science}, 2011.

\bibitem{bernstein:uist11:realtime}
M.~S. Bernstein, J.~Brandt, R.~C. Miller, and D.~R. Karger, ``Crowds in two
  seconds: Enabling realtime crowd-powered interfaces,'' in \emph{UIST 2011}.

\bibitem{bernstein:ci12}
M.~S. Bernstein, D.~R. Karger, R.~C. Miller, and J.~Brandt, ``Analytic methods
  for optimizing realtime crowdsourcing,'' \emph{arXiv 2012}.

\bibitem{patrick:wec12:crowdmanager}
P.~Minder, S.~Seuken, A.~Bernstein, and M.~Zollinger, ``Crowdmanager -
  combinatorial allocation and pricing of crowdsourcing tasks with time
  constraints,'' in \emph{ACM-EC 2012}.

\bibitem{aditya:vldb15}
Y.~Gao and A.~Parameswaran, ``Finish them!: Pricing algorithms for human
  computation,'' in \emph{VLDB 2014}.

\bibitem{ross:chi10}
J.~Ross, L.~Irani, M.~Silberman, A.~Zaldivar, and B.~Tomlinson, ``Who are the
  crowdworkers?: shifting demographics in mechanical turk,'' in \emph{CHI
  2010}.

\bibitem{basawa:book80}
I.~Basawa and B.~Rao, \emph{Statistical inference for stochastic processes},
  ser. Probability and mathematical statistics.\hskip 1em plus 0.5em minus
  0.4em\relax Academic Press, 1980.

\end{thebibliography}
\end{small}

\vspace{-1em}
%%%%%%%%%%%%%%%%%%Appendix%%%%%%%%%%%%%%%%%%%%%%%
\begin{appendix}

\section{Inference of Parameter $\lambda$}\label{section:infer}
\vspace{-0.3em}
\textbf{Fixed Period}
The likelihood function of parameter 
$\lambda$ is $L=\lambda_{N}\exp[-\lambda\sum_{k=1}^{N}T_{k}]\exp[-\lambda(T_{0}-t_{n})] 
=\lambda_{N}\exp[-\lambda T_{0}]$.
%\begin{eqnarray*}\begin{split}
%L=&\lambda_{N}\exp[-\lambda\sum_{k=1}^{N}T_{k}]\exp[-\lambda(T_{0}-t_{n})] \\
%=&\lambda_{N}\exp[-\lambda T_{0}]
%\end{split}\end{eqnarray*}
To maximize the likelihood, the ML estimation of $\lambda$ is derived as  $\hat{\lambda}=N/T_{0}$, which is unbiased according to Rao-Blackwell Theorem.

\textbf{Random Period}
Suppose each worker appears at the epochs $0<t_{1}<t_{2}\ldots<t_{N}$, we could obtain the likelihood function: $L=\lambda^{N}e^{-\lambda T_{0}}$.
%\begin{eqnarray*}\begin{split}
%L=\lambda^{N}e^{-\lambda T_{0}}
%\end{split}\end{eqnarray*}
Thus to maximize the likelihood, the ML estimation of $\lambda$ is given by  $\hat{\lambda}=\frac{N}{T_{0}}$(same as the one in Fixed Period). To remove the bias, further the parameter can be updated: $\tilde{\lambda}=((N-1)N)\hat{\lambda}$.
%\begin{eqnarray*}\begin{split}
%\tilde{\lambda}=((N-1)N)\hat{\lambda}
%\end{split}\end{eqnarray*}

\vspace{-1em}
\section{Proof of Lemma~\ref{LEMMA1}}\label{section:proof_Lemma1}
\vspace{-0.5em}
\begin{proof}
As illustrated in the previous section, Phase 1 of both $t_{1}$ and
$t_{2}$ follow exponential distribution, whose parameters are
denoted as $\lambda_{t_{1}}^{o}$ and $\lambda_{t_{2}}^{o}$, and
according to Hypothesis I in Section~\ref{section:running_para},
when allocating $t_{1}$ with $x$ unit payment and $t_{2}$ with $B-x$
unit payment, $\lambda_{t_{1}}^{o}=kx$ and
$\lambda_{t_{2}}^{o}=k(B-x)$. ($k$ is the constant coefficient) With
the two parameters defined above, we can derive the following
relationship:
$E(\{t_{1},t_{2}\}) 
=E\left\{\max\left\{L^{o}(t_{1}),L^{o} (t_{2} )\right\} \right\} 
=\frac{\lambda_{t_{1}}^{o}+\lambda_{t_{2}}^{o}}{\lambda_{t_{1}}^{o}\lambda_{t_{2}}^{o}}
- \frac{1}{\lambda_{t_{1}}^{o}+\lambda_{t_{2}}^{o}}$ %\int_{0}^{\infty}\left[\left(1-e^{-\lambda_{t_{1}}^{o} t} \right)\left(1-e^{-\lambda_{t_{2}}^{o} t}\right)\right]'t \diff t 
%\begin{eqnarray*}\begin{split}
%&E(\{t_{1},t_{2}\}) \\
%=&E\left\{\max\left\{L^{o}(t_{1}),L^{o} (t_{2} )\right\} \right\} \\
%=&\int_{0}^{\infty}\left[\left(1-e^{-\lambda_{t_{1}}^{o} t} \right)\left(1-e^{-\lambda_{t_{2}}^{o} t}\right)\right]'t \diff t \\
%=&\frac{\lambda_{t_{1}}^{o}+\lambda_{t_{2}}^{o}}{\lambda_{t_{1}}^{o}\lambda_{t_{2}}^{o}}
%- \frac{1}{\lambda_{t_{1}}^{o}+\lambda_{t_{2}}^{o}}
%\end{split}\end{eqnarray*}
As $\lambda_{t_{1}}^{o}=kx$ and $\lambda_{t_{2}}^{o}=k(B-x)$,
$E_({t_{1},t_{2}})$ turns out to be a convex function and reaches its
minimum point when $x=\frac{B}{2}$(or $\left \lfloor B/2 \right \rfloor$ if $B$ is odd). Hence, allocating $t_{1}$ and $t_{2}$ with
$\frac{B}{2}$ unit payments leads to the minimum
expected latency.
\end{proof}
\vspace{-2em}
\section{Proof of Lemma~\ref{LEMMA2}}\label{section:proof_Lemma2}
\vspace{-0.5em}
\begin{proof}
Let $\{p_{1},\ldots, p_{m}\}$ denote the payment allocated to each repetition of atomic task $t$, and $\{\lambda_{1}^{o},\ldots, \lambda_{m}^{o}\}$ denote the exponential parameter of each repetition. It is obvious to see that $\sum_{i=1}^{k}p_{i}=B$, and based on Hypothesis $1$, $\lambda_{i}^{o}=k*p_{i}$. With above information, we can derive the expected latency of $t$ as follows.\\
$E\{L(t)\}=\sum_{i=1}^{m}1/\lambda_{i}^{o}=\sum_{i=1}^{m}1/kp_{i}$.
Since $E\{L(t)\}=\sum_{i=1}^{m}1/kp_{i}\leq B^{2}/km$, and the
equality is established $iff.$ $\forall
i\in\{1,\ldots,m\}$,$x_{i}=B/m$. Therefore, we come to the
conclusion that allocating the budget evenly to each repetition
of the atomic task leads to the minimum expected latency.
\end{proof}

\vspace{-2em}
\section{Proof of Theorem~\ref{THEOREM1}}\label{section:proof_even_allocation}
\vspace{-0.5em}
\begin{proof}
This theorem is proved with mathematical induction.
Firstly, the theorem holds when both atomic tasks are run exactly once, which is a direct result of Lemma $1$.
Then, we prove that the theorem still holds when the repetitions increases to $n+1$ on condition that it hold with repetition equals to $n$.
Suppose we have a two identical tasks $t_{1}$ and $t_{2}$, which require $n$ reps. Now we allocate each rep with $x$ unit payment.Based on our presumption, this budget allocation leads to the minimum expected latency. Let $H$ denotes the completion of task $t_{1}$, and the completion of both tasks as $\max{\{H,H\}}$.
Now we increase the repetitions of both tasks to $n+1$, and the budget to $2(n+1)x$. Suppose we have a better budget allocation which outperforms allocating the budget evenly. It is trivial to see that one task will be allocated with more payment and the other atomic task will be allocated with less payment. At the same time the payment for each repetition of the same atomic task remains identical. Let $I$ denote the completion of the first $n$ repetitions of $t_{1}$ and $i$ denote the event of the completion of the last repetition of $t_{1}$. So the completion of $t_{1}$ is denoted as $\{I+i\}$.  Let $J$ denote the completion of the first $n$ repetitions of $t_{2}$ and $j$ denote the completion of the last repetition of $t_{2}$. The completion of $t_{2}$ can be denoted as ${J+j}$.
Similarly, when allocating the budget evenly to each repetition of both tasks, we use $H$ to denote the first $n$ repetitions of the atomic task and $h$ to denote the last repetition of the atomic task, and the completion of $t_{1}$ (or $t_{2}$) is denoted as $\{H+h\}$.
Then, the completion of both tasks with the assumed optimal budget allocation is denoted by $\max{\{\{I+i\},\{J+j\}\}}$ and the completion of both tasks with the evenly allocated budget is denoted as $\max{\{\{H+h\},\{H+h\}\}}$.
Here, we will have the flowing relationship: $E\{\max{\{\{I+i\},\{J+j\}\}} =(E\{I+i\}+E\{J+j\})-E\{\min{\{\{I+i\},\{J+j\}\}}\}$,
and \\
$E\{\max{\{\{H+h\},\{H+h\}\}}\}$
$ = (E\{H+h\}+E\{H+h\})-E\{\min{\{\{H+h\},\{H+h\}\}} \}$.
As $E\{I+i\}+E\{J+j\}=n/(k(x-\varepsilon))+n/(k(x+\varepsilon))$ $\leq 2n/kx=E\{H+h\}+E\{H+h\}$, we can get the result that
$ E\{\max{\{\{I+i\},\{J+j\}\}}\}\geq$ \\$E\{\max{\{\{H+h\},\{H+h\}\}}\}$.
This shows that the theorem still holds when the repetitions increases to $n+1$, which prove the theorem to be truth.
\end{proof}
\vspace{-1.5em}
\section{Proof of Lemma~\ref{LEMMA3}}\label{section:proof_Lemma3}
\vspace{-1em}
\begin{proof}
Let $r_{i}$ denote the $i$th repetition of task $x$. According to Lemma
$1$,$\forall r_{i}\in \left \{ r_{1},\ldots, r_{n}\right \}$,
$L(r_{i})$ follows exponential distribution of the same parameter
$\lambda$. So $L(x)=L(\sum_{i=1}^{k}r_{i})$. This meets the
requirement of Erlang distribution and makes $L(x)\sim
Erl(k,\lambda)$.
\end{proof}

\end{appendix}

\end{document}